\documentclass[]{elsarticle}

\usepackage{lipsum}
\makeatletter
\def\ps@pprintTitle{%
 \let\@oddhead\@empty
 \let\@evenhead\@empty
 \def\@oddfoot{}%
 \let\@evenfoot\@oddfoot}
\makeatother

\usepackage{amssymb}
\usepackage{graphicx}
\usepackage{amsfonts}
\usepackage{amsmath}
\usepackage{amsthm}
\usepackage{enumitem}
\usepackage{bm}
\usepackage[linesnumbered]{algorithm2e}
\usepackage{xcolor} 
\usepackage{hanging}
\usepackage{bbm}
\usepackage{hyperref}
\usepackage{enumitem}
\usepackage{subcaption} 
\usepackage{changepage}
\usepackage{mathtools}
\usepackage{subcaption}
\usepackage[utf8]{inputenc}
\usepackage[english]{babel} 
\usepackage{natbib} 
\setcitestyle{authoryear,open={(},close={)}}

\newtheorem{theorem}{Theorem}
\newtheorem{lemma}{Lemma}
\newtheorem{proposition}{Proposition}

\newcommand\ci{\perp\!\!\!\perp}

\begin{document}
\begin{frontmatter}

\title{Approximate Kernel-based Conditional Independence Tests for Fast Non-Parametric Causal Discovery}

\author{Eric V. Strobl, Shyam Visweswaran}
\address{Dept. of Biomedical Informatics \\ University of Pittsburgh}

\author{Kun Zhang}
\address{Dept. of Philosophy \\ Carnegie Mellon University}

\begin{abstract}
Constraint-based causal discovery (CCD) algorithms require fast and accurate conditional independence (CI) testing. The Kernel Conditional Independence Test (KCIT) is currently one of the most popular CI tests in the non-parametric setting, but many investigators cannot use KCIT with large datasets because the test scales cubicly with sample size. We therefore devise two relaxations called the Randomized Conditional Independence Test (RCIT) and the Randomized conditional Correlation Test (RCoT) which both approximate KCIT by utilizing random Fourier features. In practice, both of the proposed tests scale linearly with sample size and return accurate p-values much faster than KCIT in the large sample size context. CCD algorithms run with RCIT or RCoT also return graphs at least as accurate as the same algorithms run with KCIT but with large reductions in run time\footnote{R implementation at github.com/ericstrobl/RCIT. We recommend that users install Microsoft R Open for fast matrix computations. }.
\end{abstract}

\begin{keyword}
Conditional Independence Test \sep Non-parametric \sep Causal Discovery
\end{keyword}

\end{frontmatter}

\section{The Problem}

Constraint-based causal discovery (CCD) algorithms such as PC and FCI infer causal relations from observational data by combining the results of many conditional independence (CI) tests \citep{Spirtes00}. In practice, a CCD algorithm can easily request p-values from thousands of CI tests even with a sparse underlying graph. Developing fast and accurate CI tests is therefore critical for maximizing the usability of CCD algorithms across a wide variety of datasets.

Investigators have developed many fast parametric methods for testing CI. For example, we can use partial correlation to test for CI under the assumption of Gaussian variables \citep{Fisher15,Fisher21}. We can also consider testing for unconditional independence $X \ci Y| Z = z, \forall z$ when $Z$ is discrete and $\mathbb{P}(Z=z)>0$. The chi-squared test for instance utilizes this strategy when both $X$ and $Y$ are also discrete \citep{Pearson00}. Another permutation-based test generalizes the same strategy even when $X$ and $Y$ are not necessarily discrete \citep{Tsamardinos10}.

Testing for CI in the non-parametric setting generally demands a more sophisticated approach. One strategy involves discretizing continuous conditioning variables $Z$ as $\breve{Z}$ in some optimal fashion and assessing unconditional independence $\forall \breve{Z}=\breve{z}$ \citep{Margaritis05,Huang10}. Discretization however suffers severely from the curse of dimensionality because consistency arguments demand smaller bins with increasing sample size, but the number of cells in the associated contingency table increases exponentially with the conditioning set size. A second method involves measuring the distance between estimates of the conditional densities $f(X|Y,Z)$ and $f(X|Z)$, or their associated characteristic functions, by observing that $f(X|Y,Z)=f(X|Z)$ when $X \ci Y | Z$ \citep{Su07,Su08}. However, the power of these tests also deteriorates quickly with increases in the dimensionality of Z.

Several investigators have since proposed reproducing kernel-based CI tests in order to tame the curse of dimensionality. Indeed, kernel-based methods in general are known for their strong empirical performance in the high dimensional setting. The Kernel Conditional Independence Test (KCIT) for example assesses CI by capitalizing on a characterization of CI in reproducing kernel Hilbert spaces (RKHSs; \citep{Zhang11}). Intuitively, KCIT works by testing for vanishing regression residuals among functions in RKHSs. Another kernel-based CI test called the Permutation Conditional Independence Test (PCIT) reduces CI testing to two-sample kernel-based testing via a carefully chosen permutation found at the solution of a convex optimization problem \citep{Doran14}.

The aforementioned kernel-based CI tests unfortunately suffer from an important drawback: both tests scale at least quadratically with sample size and therefore take too long to return a p-value in the large sample size setting. In particular, KCIT's bottleneck lies in the eigendecomposition as well as the inversion of large kernel matrices \citep{Zhang11}, and PCIT takes too long to solve for its required permutation \citep{Doran14}. As a general rule, it is difficult to develop exact kernel-based methods which scale sub-quadratically with sample size, since the computation of kernel matrices themselves scales at least quadratically.

Many investigators have nonetheless utilized \textit{random Fourier features} in order to quickly approximate kernel methods. For example, Lopez-Paz and colleagues developed an unconditional independence test using statistics obtained from canonical correlation analysis with random Fourier features \citep{LopezPaz13}. Others have analyzed the use of random Fourier features for predictive modeling (e.g., \citep{Rahimi07, Sutherland15}) or dimensionality reduction \citep{LopezPaz14}. In practice, investigators have observed that methods which utilize random Fourier features often scale linearly with sample size and achieve comparable accuracy to exact kernel methods.

In this paper, we also use random Fourier features to design two fast tests called the Randomized Conditional Independence Test (RCIT) and the Randomized conditional Correlation Test (RCoT) which approximate the solution of KCIT. Simulations show that RCIT, RCoT and KCIT have comparable accuracy, but both RCIT and RCoT scale linearly with sample size in practice. As a result, RCIT and RCoT return p-values several orders of magnitude faster than KCIT in the large sample size context. Moreover, experiments demonstrate that the causal structures returned by CCD algorithms using either RCIT, RCoT or KCIT have nearly identical accuracy.

\section{Characterizations of Conditional Independence}

Capital letters $X,Y,Z$ denote sets of random variables with domains $\mathcal{X}, \mathcal{Y}, \mathcal{Z}$, respectively. Consider a measurable, positive definite kernel $k_{\mathcal{X}}$ on $\mathcal{X}$ and denote the corresponding RKHS by $\mathcal{H}_{\mathcal{X}}$. We similarly define $k_{\mathcal{Y}}$, $\mathcal{H}_{\mathcal{Y}}$, $k_{\mathcal{Z}}$, and $\mathcal{H}_{\mathcal{Z}}$. We denote the probability distribution of $X$ as $\mathbb{P}_X$ and the joint probability distribution of $(X,Z)$ as $\mathbb{P}_{XZ}$. Let $L^2_X$ denote the space of square integrable functions of $X$, and $L^2_{XZ}$ that of $(X, Z)$. Here, $L^2_X = \{s(X) \mid \mathbb{E}_X(|s|^2) < \infty \}$ and likewise for $L^2_{XZ}$. Next consider a dataset of $n$ i.i.d. samples drawn according to $\mathbb{P}_{XYZ}$.

We use the notation $X \ci Y | Z$ when $X$ and $Y$ are conditionally independent given $Z$. Perhaps the simplest characterization of CI reads as follows: $X \ci Y | Z$ if and only if $\mathbb{P}_{XY|Z}=\mathbb{P}_{X|Z}\mathbb{P}_{Y|Z}$. Equivalently, we have $\mathbb{P}_{X|YZ}=\mathbb{P}_{X|Z}$ and $\mathbb{P}_{Y|XZ}=\mathbb{P}_{Y|Z}$. 

\subsection{Characterization by RKHSs}
A second characterization of CI is given in terms of the cross-covariance operator $\Sigma_{XY}$ on RKHSs \citep{Fukumizu04}. For the random vector $(X, Y )$ on $\mathcal{X} \times \mathcal{Y}$, we define the cross-covariance operator from $\mathcal{H}_{\mathcal{Y}}$ to $\mathcal{H}_{\mathcal{X}}$ as follows:
\begin{equation}
\langle f, \Sigma_{XY} g \rangle = \mathbb{E}_{XY} [f(X)g(Y)] - \mathbb{E}_X[f(X)]\mathbb{E}_Y [g(Y)]
\end{equation}
for all $f \in \mathcal{H}_{\mathcal{X}}$ and $g \in \mathcal{H}_{\mathcal{Y}}$. We may then define the partial cross-covariance operator of $(X, Y)$ given $Z$ by\footnote{Use the right inverse instead of the inverse, if $\Sigma_{ZZ}$ is not invertible (see Corollary 3 in \citep{Fukumizu04}).}:
\begin{equation} \label{ccc_op}
\Sigma_{XY \cdot Z} = \Sigma_{XY} - \Sigma_{XZ}\Sigma_{ZZ}^{-1} 	\Sigma_{ZY}.
\end{equation}
Notice the similarity of the partial cross-covariance operator to the linear partial cross-covariance matrix (as well as the conditional cross-covariance matrix in the Gaussian case)\footnote{Recall that the partial cross-covariance of $X$ and $Y$ given $Z$ is defined as $\mathbb{E}[(X-\mathbb{E}(X|Z))(Y-\mathbb{E}(Y|Z))]$; in other words, it is equivalent to the cross-covariance of $X$ and $Y$ given $Z$. In contrast, the conditional cross-covariance of $X$ and $Y$ given $Z$ is defined as $\mathbb{E}[(X-\mathbb{E}(X|Z))(Y-\mathbb{E}(Y|Z))|Z]$ (notice the extra conditioning).}. Intuitively, one can interpret the above equation as the partial covariance between $\{f(X), \forall f \in \mathcal{H}_{\mathcal{X}} \}$ and $\{g(Y), \forall g \in \mathcal{H}_{\mathcal{Y}} \}$ given $\{h(Z), \forall h \in \mathcal{H}_{\mathcal{Z}} \}$.

Now if we use characteristic kernels\footnote{A kernel $k_X$ is characteristic if
$\mathbb{E}_{X \sim \mathbb{P}_X} [f(X)] = \mathbb{E}_{X \sim \mathbb{Q}_X}[f(X)], \forall f \in \mathcal{H}_{\mathcal{X}}$ implies $\mathbb{P}_X = \mathbb{Q}_X$, where $\mathbb{P}_X$ and $\mathbb{Q}_X$ are two probability distributions of $X$ \citep{Fukumizu08}. Two examples of characteristic kernels include the Gaussian RBF kernel and the Laplacian kernel. } in \eqref{ccc_op}, then the partial cross-covariance operator is related to the CI relation via the following conclusion:
\begin{proposition} \label{prop_RKHS}
\citep{Fukumizu04,Fukumizu08} Let $\ddot{X} = (X,Z)$ and $k_{\ddot{\mathcal{X}}}=k_{\mathcal{X}} k_{\mathcal{Z}}$. Also let $\mathcal{H}_{\ddot{\mathcal{X}}}$ represent the RKHS corresponding to $k_{\ddot{\mathcal{X}}}$. Assume $\mathbb{E}[k_{\mathcal{X}}(X,X)]<\infty$ and $\mathbb{E}[k_{\mathcal{Y}}(Y,Y)]<\infty$\footnote{This assumption ensures that $\mathcal{H}_{\mathcal{X}} \subset L^2_X$ and $\mathcal{H}_{\mathcal{Y}} \subset L^2_Y$.}. Further assume that $k_{\ddot{\mathcal{X}}}k_{\mathcal{Y}}$ is a characteristic kernel on $(\mathcal{X} \times \mathcal{Y}) \times \mathcal{Z}$, and that $\mathcal{H}_{\mathcal{Z}} + \mathbb{R}$ (the direct sum of the two RKHSs) is dense in $L^2_Z$. Then
\begin{equation}
\Sigma_{\ddot{X}Y \cdot Z} = 0 \iff X \ci Y |Z.
\end{equation}

\end{proposition}

\subsection{Characterization by $L^2$ spaces}
We also consider a different characterization of CI which enforces the uncorrelatedness of functions in
suitable spaces; this definition is intuitively more appealing. In particular, consider the following constrained $L^2$ spaces:
\begin{equation}
\begin{aligned}
&\mathcal{F}_{XZ} \triangleq \{
f \in L^2_{XZ} \mid \mathbb{E}(f|Z) = 0 \}, \\
&\mathcal{F}_{YZ} \triangleq \{
g \in L^2_{YZ} \mid \mathbb{E}(g|Z) = 0 \}, \\
&\mathcal{F}_{Y \cdot Z} \triangleq \{
h \mid h=h^{\prime}(Y) - \mathbb{E}(h^{\prime}|Z) , h^{\prime} \in L_Y^2 \}.
\end{aligned}
\end{equation}
We then have the following result:
\begin{proposition} \label{prop_L2} 
\citep{Daudin80} The following conditions are equivalent:
\begin{enumerate}
\item $X \ci Y|Z$,
\item $\mathbb{E}(fg) = 0, \forall f \in \mathcal{F}_{XZ} \text{ \textnormal{and} } \forall g \in \mathcal{F}_{YZ}$,
\item $\mathbb{E}(f g^{\prime}) = 0, \forall f \in \mathcal{F}_{XZ} \text{ \textnormal{and} } \forall g^{\prime} \in L^2_{YZ}$,
\item $\mathbb{E}(f h) = 0, \forall f \in \mathcal{F}_{XZ} \text{ \textnormal{and} } \forall h \in \mathcal{F}_{Y \cdot Z}$,
\item $\mathbb{E}(f h^{\prime}) = 0, \forall f \in \mathcal{F}_{XZ} \text{ \textnormal{and} } \forall h^{\prime} \in L_{Y}^2$.
\end{enumerate}
\end{proposition}
\noindent The second condition means that any ``residual'' function of $(X, Z)$ given $Z$ is uncorrelated with that of $(Y, Z)$ given $Z$. The equivalence also represents a generalization of the case when $(X, Y, Z)$ is jointly Gaussian; here, $X \ci Y | Z$ if and only if any residual function of $X$ given $Z$ is uncorrelated with that of $Y$ given $Z$; i.e., the linear partial correlation coefficient $\rho_{XY \cdot Z}$ is zero.

We also encourage the reader to observe the close relationship between Proposition \ref{prop_RKHS} and claim 4 of Proposition \ref{prop_L2}. Here, we have almost equivalent statements, but Proposition \ref{prop_RKHS} only considers functions in RKHSs, while claim 4 of Proposition \ref{prop_L2} considers functions in $L^2$ spaces. We find Proposition \ref{prop_RKHS} more useful than claim 4 of Proposition \ref{prop_L2} because the RKHS of a characteristic kernel might be much smaller than the corresponding $L^2$ space. 

\section{Test Statistic \& its Asymptotic Distribution}

We consider the following hypotheses:
\begin{equation} \label{hypo}
\begin{aligned}
&H_0: X \ci Y | Z, \\
&H_1: X \not \ci Y | Z.
\end{aligned}
\end{equation}

Now KCIT uses an empirical estimate of the squared Hilbert-Schmidt norm of the partial cross-covariance operator as a statistic to determine whether to reject $H_0$:
\begin{equation} \label{KCIT_stat}
\mathcal{S}_{K} = n \| \Sigma_{\ddot{X}Y \cdot Z} \|_{\widehat{HS}}^2.
\end{equation}
Here, $\| \Sigma_{\ddot{X}Y \cdot Z} \|_{\widehat{HS}}^2$ denotes an empirical estimate of $\| \Sigma_{\ddot{X}Y \cdot Z} \|_{HS}^2$, which we can compute using centered kernel matrices (see Theorem 4 and Proposition 5 of \citep{Zhang11} for details). We can justify $\mathcal{S}_{K}$ as a measure of CI due to Proposition \ref{prop_RKHS}. We may thus equivalently rewrite the null and alternative in \ref{hypo} more explicitly as follows:
\begin{equation} \label{hypo1}
\begin{aligned}
&H_0: \| \Sigma_{\ddot{X}Y \cdot Z} \|_{HS}^2 =0, \\
&H_1: \| \Sigma_{\ddot{X}Y \cdot Z} \|_{HS}^2 > 0.
\end{aligned}
\end{equation}

In this report, we will also take advantage of the characterization of CI presented in Proposition \ref{prop_RKHS}. Recall that the Frobenius norm corresponds to the Hilbert-Schmidt norm in Euclidean space. We therefore consider the squared Frobenius norm of the empirical partial cross-covariance matrix as an approximation of \ref{KCIT_stat} for RCIT:
\begin{equation}
\mathcal{S} = n\| \widehat{\Sigma}_{\ddot{A}B \cdot C} \|_F^2,
\end{equation}
where  $\widehat{\Sigma}_{\ddot{A}B \cdot C} = \frac{1}{n-1} \sum_{i=1}^n[(\ddot{A}_i-\widehat{\mathbb{E}}(\ddot{A}|C))(B_i-\widehat{\mathbb{E}}(B|C))]$ resembles the empirical cross-covariance matrix. We also have $\ddot{A} = f^{\prime}(\ddot{X}) \triangleq \{f_1^{\prime}(\ddot{X}),$ $\dots,f_m^{\prime}(\ddot{X}) \}$ with $f_j^{\prime}(\ddot{X}) \in \mathcal{G}_{\ddot{{\mathcal{X}}}}, \forall j$. Similarly, $B = h^{\prime}(Y) \triangleq \{h_1^{\prime}(Y),$ $\dots,h_q^{\prime}(Y) \}$ with $h_k^{\prime}(Y) \in \mathcal{G}_{\mathcal{Y}}, \forall k$, and $C = g(Z) \triangleq \{g_1(Z),\dots,g_d(Z) \}$ with $g_l(Z) \in \mathcal{G}_{\mathcal{Z}}, \forall l$. Here, $\mathcal{G}_{\ddot{\mathcal{X}}}$, $\mathcal{G}_{\mathcal{Y}}$, and $\mathcal{G}_{\mathcal{Z}}$ denote three spaces of functions, which we will specify shortly. In other words, we select $m$ functions from $\mathcal{G}_{\ddot{\mathcal{X}}}$, $q$ functions from $\mathcal{G}_{\mathcal{Y}}$, and $d$ functions from $\mathcal{G}_{\mathcal{Z}}$. We henceforth choose to take the following hypotheses as equivalent to those in \ref{hypo1} and \ref{hypo}:
\begin{equation} \label{hypo2}
\begin{aligned}
&H_0: \| \Sigma_{\ddot{A}B \cdot C} \|_F^2 =0, \\
&H_1: \| \Sigma_{\ddot{A}B \cdot C} \|_F^2 > 0.
\end{aligned}
\end{equation}

Now we will compute $\widehat{\Sigma}_{\ddot{A}B \cdot C}$ using $\widehat{\Sigma}_{\ddot{A}B} - \widehat{\Sigma}_{\ddot{A}C}(\widehat{\Sigma}_{CC}+\gamma I)^{-1} \widehat{\Sigma}_{CB}$ similar to \ref{ccc_op}, where $\gamma$ denotes a small ridge parameter; recall that this is equivalent to computing the cross-covariance matrix across the residuals of $\ddot{A}$ and $B$ given $C$ using linear ridge regression. We thus may \textit{not} necessarily have $\widehat{\Sigma}_{\ddot{A}B \cdot C} = \widehat{\Sigma}_{\ddot{A}B} - \widehat{\Sigma}_{\ddot{A}C}(\widehat{\Sigma}_{CC}+\gamma I)^{-1} \widehat{\Sigma}_{CB}$. However, we may have $\widehat{\Sigma}_{\ddot{A}B \cdot C} \approx \widehat{\Sigma}_{\ddot{A}B} - \widehat{\Sigma}_{\ddot{A}C}(\widehat{\Sigma}_{CC}+\gamma I)^{-1} \widehat{\Sigma}_{CB}$, if we choose $\mathcal{G}_{\mathcal{Z}}$ in the right way. We therefore must define the space $\mathcal{G}_{\mathcal{Z}}$ in a sensible manner.

In this report, we will set $\mathcal{G}_{\mathcal{Z}}$ to $\{ \sqrt{2} \text{cos}(W^T Z + B) | W \sim \mathbb{P}_W,$ $B \sim $ \\$\text{Uniform}([0, 2\pi]) \}$ and likewise for $\mathcal{G}_{\ddot{\mathcal{X}}}$ and $\mathcal{G}_{\mathcal{Y}}$. We select these specific spaces because we can use them to approximate continuous shift-invariant kernels\footnote{A kernel $k$ is said to be shift-invariant if and only if, for any $a \in \mathbb{R}^p$, we have $k(x-a,y-a) = k(x,y)$, $\forall (x,y) \in \mathbb{R}^p \times \mathbb{R}^p$.}, such as the Gaussian RBF kernel or the Laplacian kernel, via the following result:
\begin{proposition} \citep{Rahimi07} For a continuous shift-invariant kernel $k(x, y)$ on $\mathbb{R}^p$, we have:
\begin{equation}
k(x, y) = \int_{\mathbb{R}^p} e^{iw^T(x-y)} ~dF_w
= \mathbb{E}[\zeta(x)\zeta(y)],
\end{equation}
where $F_W$ represents the CDF of $\mathbb{P}_W$ and $\zeta(x) =$ $\sqrt{2}\text{\textnormal{cos}}(W^Tx+B)$ with $W \sim \mathbb{P}_W$ and $B \sim \text{\textnormal{Uniform}}([0, 2\pi])$.
\end{proposition}
\noindent The precise form of $\mathbb{P}_W$ depends on the type of shift-invariant kernel one would like to approximate (see Figure 1 of \citep{Rahimi07} for a list). Since KCIT uses the Gaussian RBF kernel, we choose to approximate the Gaussian RBF kernel by setting $\mathbb{P}_W$ to a Gaussian.

Now let $f_j = f_j^{\prime} - \mathbb{E}(f_j^{\prime}|Z)$. Then $\mathbb{E}( f_j | Z )=0$, so $f_j \in \mathcal{F}_{XZ}$. Moreover, $h_k^{\prime} - \mathbb{E}(h_k^{\prime}|Z) \in \mathcal{F}_{Y  \cdot Z}$. Note that we can estimate $\mathbb{E}(f_j^{\prime}|Z)$ with the linear ridge regression solution $\widehat{u}_j^T g(Z)$ under mild conditions because we can guarantee that $\mathbb{P}\big[ | \widehat{\mathbb{E}}_R(f_j^{\prime}|Z) - \widehat{u}_j^T g(Z) | \geq \varepsilon \big] \rightarrow 0$ for any fixed $\varepsilon >0$, where $\widehat{\mathbb{E}}_R(f_j^{\prime}|Z)$ denotes the estimate of $\mathbb{E}(f_j^{\prime}|Z)$ by kernel ridge regression; this holds so long as we choose $d$ large enough for $g(Z)$ (see Section 3.1 of \citep{Sutherland15}; the argument is complex and beyond the scope of this paper). We can also estimate $\mathbb{E}(h_k^{\prime}|Z)$ with $\widehat{u}_k^T g(Z)$, because we can similarly guarantee that $\mathbb{P}\big[ | \widehat{\mathbb{E}}_R(h_k^{\prime}|Z) - \widehat{u}_k^T g(Z) | \geq \varepsilon \big] \rightarrow 0$ for any fixed $\varepsilon>0$. 

We can therefore consider the following spaces for $\mathcal{S}$ which are similar to the $L^2$ spaces used in claim 4 of Proposition \ref{prop_L2}:
\begin{equation}
\begin{aligned}
&\widehat{\mathcal{G}}_{\ddot{X}} \triangleq \big\{
f \mid f_j = f_j^{\prime} - \mathbb{E}(f_j^{\prime}|Z), f_j^{\prime} \in \mathcal{G}_{\ddot{{\mathcal{X}}}} \big\}, \\
&\widehat{\mathcal{G}}_{Y \cdot Z} \triangleq \big\{
h \mid h_k = h_k^{\prime}-\mathbb{E}(h_k^{\prime}|Z),  h_k^{\prime} \in \mathcal{G}_{\mathcal{Y}} \big\}.
\end{aligned}
\end{equation}
\noindent We then approximate CI with $\mathcal{S}$ in the following sense:
\begin{enumerate}
\item We always have $X \ci Y|Z \implies \mathbb{E}(fh) = 0, \forall f \in \widehat{\mathcal{G}}_{\ddot{X}} \text{ \textnormal{and} } \forall h \in \widehat{\mathcal{G}}_{Y  \cdot Z}$.
\item The reverse direction will hold for an increasing number of distributions as $m,q$ increase.
\end{enumerate}
\noindent In practice, we find that the second point holds in all of the cases we tested with only $m,q=5$.

\subsection{Null Distribution}
We now consider the asymptotic distribution of $\mathcal{S}$ under the null. 

\begin{theorem} \label{thm_asymp}
Consider $n$ i.i.d. samples from $\mathbb{P}_{XYZ}$. We then have the following asymptotic distribution under the null in \ref{hypo2}:
\begin{equation}
n\| \widehat{\Sigma}_{\ddot{A}B  \cdot C} \|_F^2 \stackrel{d}{\rightarrow} \sum_{i=1}^{L} \lambda_i z_i^2,
\end{equation}
where $\{z_1,\dots,z_L\}$ denotes i.i.d. standard Gaussian variables (thus $\{z_1^2,\dots,z_L^2\}$ denotes i.i.d $\chi_1^2$ variables), $L$ the number of elements in $\widehat{\Sigma}_{\ddot{A}B \cdot C}$, and $\lambda$ the eigenvalues of the covariance matrix $\Pi$, which we assume to be positive definite; the matrix $\Pi$ is more specifically the covariance matrix of the vectorization of $(\ddot{A} - \mathbb{E}(\ddot{A}|C)) (B - \mathbb{E}(B|C))^T$. We may denote an arbitrary entry in $\Pi$ as follows:
\begin{equation} \label{cov_chi2}
\begin{aligned}
&\hspace{4.5mm} \Pi_{\ddot{A}_iB_j,\ddot{A}_kB_l} \\
&=\mathbb{E} \big[ (\ddot{A}_i - \mathbb{E}(\ddot{A}_i|C))(B_j - \mathbb{E}(B_j|C))(\ddot{A}_k - \mathbb{E}(A_k|C))(B_l - \mathbb{E}(B_l|C) )\big].
\end{aligned}
\end{equation}
\end{theorem}
\begin{proof}
We may first write:
\begin{equation}
\begin{aligned}
& \hspace{4.5mm} n\| \widehat{\Sigma}_{\ddot{A}B \cdot C} \|_F^2 \\
& = n*\text{tr}(\widehat{\Sigma}_{\ddot{A}B \cdot C}\widehat{\Sigma}_{\ddot{A}B \cdot C}^T) \\
& = n*v(\widehat{\Sigma}_{\ddot{A}B \cdot C})^T v(\widehat{\Sigma}_{\ddot{A}B \cdot C}), \\
& = \big[\sqrt{n}v(\widehat{\Sigma}_{\ddot{A}B \cdot C})\big]^T \big[\sqrt{n}v(\widehat{\Sigma}_{\ddot{A}B \cdot C}) \big],
\end{aligned}
\end{equation}
where $v(\widehat{\Sigma}_{\ddot{A}B \cdot C})$ stands for the vectorization of $\widehat{\Sigma}_{\ddot{A}B \cdot C}$. By CLT of the sample covariance matrix (see Lemma \ref{cov_CLT} in the Appendix) combined with the continuous mapping theorem and the null, we know that $\sqrt{n}v(\widehat{\Sigma}_{\ddot{A}B \cdot C}) \stackrel{d}{\rightarrow} \mathcal{N}(0,\Pi)$. Here, we write an arbitrary entry $\Pi_{\ddot{A}_iB_j,\ddot{A}_kB_l}$ under the null as follows:
\begin{equation} \label{deriv1}
\begin{aligned}
&\hspace{4.5mm} \Pi_{\ddot{A}_iB_j,\ddot{A}_kB_l}\\
&= \text{\normalfont{Cov}} \big[ (\ddot{A}_i - \mathbb{E}(\ddot{A}_i|C))(B_j - \mathbb{E}(B_j|C)),\\
&\hspace{45mm} (\ddot{A}_k - \mathbb{E}(\ddot{A}_k|C))(B_l - \mathbb{E}(B_l|C) ) \big] \\
& = \mathbb{E} \big[ (\ddot{A}_i - \mathbb{E}(\ddot{A}_i|C))(B_j - \mathbb{E}(B_j|C))*\\
&\hspace{45mm} (\ddot{A}_k - \mathbb{E}(\ddot{A}_k|C))(B_l - \mathbb{E}(B_l|C) ) \big].
\end{aligned}
\end{equation}

Now consider the eigendecomposition of $\Pi$ written as $\Pi=E\Lambda E^T$. Then, we have $E^T \big[\sqrt{n}v(\widehat{\Sigma}_{\ddot{A}B \cdot C})\big]\stackrel{d}{\rightarrow} \mathcal{N}(0,\Lambda)$ by the continuous mapping theorem. Note that:
\begin{equation}
\begin{aligned}
& \hspace{4.5mm} \big[\sqrt{n}v(\widehat{\Sigma}_{\ddot{A}B \cdot C})\big]^T \big[\sqrt{n}v(\widehat{\Sigma}_{\ddot{A}B \cdot C}) \big] \\
& = \big(E^T \big[\sqrt{n}v(\widehat{\Sigma}_{\ddot{A}B \cdot C})\big] \big)^T \big(E^T \big[\sqrt{n}v(\widehat{\Sigma}_{\ddot{A}B \cdot C})\big] \big) \\
& \stackrel{d}{\rightarrow} \sum_{i=1}^{L} \lambda_i z_i^2.
\end{aligned}
\end{equation}

\end{proof}

\noindent We conclude that the null distribution of the test statistic is a positively weighted sum of i.i.d. $\chi_1^2$ random variables. Note that we can obtain estimates of the conditional expectations in $\Pi$ by using kernel ridge regressions. We will however not need to perform the kernel ridge regressions directly, because we can approximate the outputs of kernel ridge regressions to within an arbitrary degree of accuracy using linear ridge regressions with enough random Fourier features \citep{Sutherland15}. We can finally obtain an estimate of $\Pi$ by application of the continuous mapping theorem and the weak law of large numbers. For an arbitrary entry in $\Pi$:
\begin{equation}
\begin{aligned}
& \hspace{5mm} \frac{1}{n} \sum_{r=1}^n (\ddot{A}_{i,r} - \widehat{\mathbb{E}}(\ddot{A}_i|C))(B_{j,r} - \widehat{\mathbb{E}}(B_j|C))* \\ & \hspace{45mm} (\ddot{A}_{k,r} - \widehat{\mathbb{E}}(\ddot{A}_k|C))(B_{l,r} - \widehat{\mathbb{E}}(B_l|C))\\
& \stackrel{p}{\rightarrow} \mathbb{E} \big[ (\ddot{A}_i - \mathbb{E}(\ddot{A}_i|C))(B_j - \mathbb{E}(B_j|C))* \\ & \hspace{45mm} (\ddot{A}_k - \mathbb{E}(\ddot{A}_k|C))(B_l - \mathbb{E}(B_l|C) )\big].
\end{aligned}
\end{equation}

Unfortunately, a closed form CDF of a positively weighted sum of chi-squared random variables does not exist in general. We can approximate the CDF by Imhof's method which inverts the characteristic function numerically \citep{Imhof61}. We should consider Imhof's method as exact, since it provides error bounds and can be used to compute the distribution at a fixed point to within a desired precision \citep{Solomon77, Johnson02}. However, Imhof's method is too computationally intensive for our purposes. We can nonetheless utilize several fast methods which approximate the null by moment matching.

\subsection{Approximating the Null Distribution by Moment Matching}
We write the cumulants of a positively weighted sum of i.i.d. $\chi_1^2$ random variables as follows:
\begin{equation}
c_r = 2^{r-1}(r-1)!\sum_{i=1}^L \lambda_i^r,
\end{equation}
where $\lambda = \{\lambda_1, \dots, \lambda_L\}$ denotes the weights. We may for example derive the first three cumulants:
\begin{equation}
m_1 = \sum_{i=1}^L \lambda_i, ~~~ m_2 = 2\sum_{i=1}^L \lambda_i^2, ~~~ m_3 = 8\sum_{i=1}^L \lambda_i^3.
\end{equation}
We then recover the moments from the cumulants as follows:
\begin{equation}
m_r = c_r + \sum_{i=1}^{r-1}{r-1 \choose i-1}c_i m_{r-i}, ~~~r=2,3,\dots
\end{equation}

Now the Satterthwaite-Welch method \citep{Welch38,Satterthwaite46,Fairfield36} represents perhaps the simplest and earliest moment matching method. The method matches the first two moments of the sum with a gamma distribution $\Gamma(\widehat{g},\widehat{\theta})$. Zhang and colleagues adopted a similar strategy in their paper introducing KCIT \citep{Zhang11}. Here, we have:
\begin{equation}
\widehat{g} = \frac{1}{2} c_1^2/c_2, ~~~ \widehat{\theta} = c_2/c_1.
\end{equation}

We however find the above gamma approximation rather crude. We therefore also consider applying more modern methods to estimating the distribution of a sum of positively weighted chi-squares. Improved methods such as the Hall-Buckley-Eagleson \citep{Hall83,Buckley88} and the Wood F \citep{Wood89} methods match the first three moments of the sum to other distributions in a similar fashion. On the other hand, the Lindsay-Pilla-Basak method \citep{Lindsay00} matches the first $2L$ moments to a mixture distribution. 

We will focus on the Lindsay-Pilla-Basak method in this paper, since Bodenham \& Adams have already determined that the Lindsay-Pilla-Basak method performs the best through extensive experimentation \citep{Bodenham16,Bodenham15}. We therefore choose to use the method as the default method for RCIT. Briefly, the method approximates the CDF under the null $F_{\mathcal{H}_0}$ using a finite mixture of $L$ Gamma CDFs $F_{\Gamma(g,\theta_i)}$:
\begin{equation}
F_{\mathcal{H}_0} = \sum_{i=1}^L \pi_i F_{\Gamma(g,\theta_i)},
\end{equation}
where $\pi \geq 0, \sum_{i=1}^L \pi_i =1$, and we seek to determine the $2L+1$ parameters $g$, $\theta_1, \dots, \theta_L$, and $\pi_1, \dots, \pi_L$. The Lindsay-Pilla-Basak method computes these parameters by a specific sequence of steps that makes use of results concerning moment matrices (see Appendix II in \citep{Uspensky37}). The sequence is complicated and beyond the scope of this paper, but we refer the reader to \citep{Lindsay00} for details.

\subsection{Testing for Conditional Un-Correlatedness}

Strictly speaking, we must consider the extended variable set $\ddot{X}$ to test for conditional independence according to Proposition \ref{prop_RKHS}. However, we have two observations: (1) we can substitute a test for non-linear conditional uncorrelatedness with tests for conditional independence in almost all cases encountered in practice because most conditionally dependent variables are correlated after some functional transformations, and (2) using the extended variable set $\ddot{X}$ makes estimating the null distribution more difficult compared to using the unextended variable set $X$. The first observation coincides with the observations of others who have noticed that Fisher's z-test performs well (but not perfectly) in ruling out conditional independencies with non-Gaussian data. We can also justify the first observation with the following result using the cross-covariance operator $\Sigma_{XY \cdot Z}$:
\begin{proposition} \label{prop_RKHS2}
\citep{Fukumizu04,Fukumizu08} Assume $\mathbb{E}[k_{\mathcal{X}}(X,X)]<\infty$ and $\mathbb{E}[k_{\mathcal{Y}}(Y,Y)]<\infty$. Further assume that $k_{\mathcal{X}}k_{\mathcal{Y}}$ is a characteristic kernel on $\mathcal{X} \times \mathcal{Y}$, and that $\mathcal{H}_{\mathcal{Z}} + \mathbb{R}$ (the direct sum of the two RKHSs) is dense in $L^2_Z$. Then
\begin{equation}
\Sigma_{XY \cdot Z} = 0 \iff \mathbb{E}_Z\big[ \mathbb{P}_{XY|Z} \big]=\mathbb{E}_Z\big[ \mathbb{P}_{X|Z}\mathbb{P}_{Y|Z} \big].
\end{equation}

\end{proposition}
\noindent In other words, we have:
\begin{equation}
\begin{aligned}
& \Sigma_{XY \cdot Z} = 0 \implies \mathbb{P}_{XY}=\int \mathbb{P}_{X|Z}\mathbb{P}_{Y|Z} ~ d\mathbb{P}_{Z}, \\
& \Sigma_{XY \cdot Z} = 0 \impliedby \mathbb{E}_Z\big[ \mathbb{P}_{X|Z}\mathbb{P}_{Y|Z} \big]=\mathbb{E}_Z\big[ \mathbb{P}_{XY|Z} \big] \impliedby X \ci Y |Z.
\end{aligned}
\end{equation}
Notice that $\Sigma_{XY \cdot Z}=0$ is almost equivalent to CI, in the sense that $\Sigma_{XY \cdot Z}=0$ just misses those rather contrived distributions where $\mathbb{P}_{XY}=\int \mathbb{P}_{XY|Z} ~ d\mathbb{P}_{Z}=\int \mathbb{P}_{X|Z}\mathbb{P}_{Y|Z} ~ d\mathbb{P}_Z$ when $X \not \ci Y | Z$. In other words, if $\mathbb{P}_{XY} \not =\int \mathbb{P}_{X|Z}\mathbb{P}_{Y|Z} ~ d\mathbb{P}_{Z}$ when $X \not \ci Y |Z$, then we have $\Sigma_{XY \cdot Z} = 0 \iff \Sigma_{\ddot{X}Y \cdot Z} = 0$ (under the corresponding additional assumptions of Propositions \ref{prop_RKHS} and \ref{prop_RKHS2}).

Let us now consider an example of a situation where $\int \mathbb{P}_{XY|Z} ~ d\mathbb{P}_{Z}$ $\not =$\\ $\int \mathbb{P}_{X|Z}\mathbb{P}_{Y|Z} ~ d\mathbb{P}_{Z}$ when $X \not \ci Y |Z$. Take three binary variables $X,Y,Z \in \{0,1\}$. Let $\mathbb{P}_{Z=0} = 0.2$ and $\mathbb{P}_{Z=1} = 0.8$. Also consider the four probability tables in Table \ref{table_example}.
\begin{table}
\begin{subtable}{0.5\textwidth}
  \centering
\begin{tabular}{ l | c | r}
  ${}$ & $\mathbb{P}_{X|Z=0}$ & $\mathbb{P}_{X|Z=1}$ \\
	\hline
  $X=0$ & 0.5 & 0.3\\
  $X=1$ & 0.5 & 0.7\\
\end{tabular}
\caption{} \label{table1a}
\end{subtable}
\begin{subtable}{0.5\textwidth}
  \centering
\begin{tabular}{ l | c | r}
  ${}$ & $\mathbb{P}_{Y|Z=0}$ & $\mathbb{P}_{Y|Z=1}$ \\
	\hline
  $Y=0$ & 0.3 & 0.4\\
  $Y=1$ & 0.7 & 0.6\\
\end{tabular}
\caption{} \label{table1b}
\end{subtable}
\begin{subtable}{0.5\textwidth}
  \centering
\begin{tabular}{ l | c | r}
  ${}$ & $\mathbb{P}_{XY|Z=0}$ & $\mathbb{P}_{XY|Z=1}$ \\
	\hline
  $X=0, Y=0$ & 0.2 & 0.1075\\
  $X=0, Y=1$ & 0.3 & 0.1925\\
  $X=1, Y=0$ & 0.1 & 0.2925\\
  $X=1, Y=1$ & 0.4 & 0.4075\\
\end{tabular}
\caption{} \label{table1c}
\end{subtable}
\begin{subtable}{0.5\textwidth}
  \centering
\begin{tabular}{ l | c}
  ${}$ & $\mathbb{P}_{XY}$ \\
	\hline
  $X=0, Y=0$ & 0.126 \\
  $X=0, Y=1$ & 0.214 \\
  $X=1, Y=0$ & 0.254 \\
  $X=1, Y=1$ & 0.406 \\
\end{tabular}
\caption{} \label{table1d}
\end{subtable}

\caption{Example of a situation where $\int \mathbb{P}_{XY|Z} ~ d\mathbb{P}_{Z} =  \int \mathbb{P}_{X|Z}{P}_{Y|Z} ~ d\mathbb{P}_{Z}$ when $X \not \ci Y | Z$ using binary variables.} \label{table_example}
\end{table}
Here, we have chosen the probabilities in the tables carefully by satisfying the following equation:
\begin{equation} \label{solve_eq}
\begin{aligned}
&\int \mathbb{P}_{XY|Z} ~ d\mathbb{P}_{Z} =  \int \mathbb{P}_{X|Z}{P}_{Y|Z} ~ d\mathbb{P}_{Z}\\
\iff &\mathbb{P}_{Z=0} (\mathbb{P}_{X|Z=0}\mathbb{P}_{Y|Z=0}) + \mathbb{P}_{Z=1} (\mathbb{P}_{X|Z=1}\mathbb{P}_{Y|Z=1})\\
& = \mathbb{P}_{Z=0} \mathbb{P}_{XY|Z=0} + \mathbb{P}_{Z=1} \mathbb{P}_{XY|Z=1}.
\end{aligned}
\end{equation}
Of course, the equality holds when we have conditional independence $\mathbb{P}_{XY|Z} = \mathbb{P}_{X|Z} \mathbb{P}_{Y|Z}$. We are however interested in the case when conditional dependence holds. We therefore instantiated the values of Tables \ref{table1a} and \ref{table1b} as well as the second column in Table \ref{table1c} ($\mathbb{P}_{XY|Z=0}$) such that $\mathbb{P}_{XY|Z=0} \not = \mathbb{P}_{X|Z=0} \mathbb{P}_{Y|Z=0}$. We then solved for $\mathbb{P}_{XY|Z=1}$ using Equation \ref{solve_eq} in order to complete Table \ref{table1c}. This ultimately yielded Table \ref{table1d}.

Notice that we obtain a unique value for $\mathbb{P}_{XY|Z=1}$ by solving Equation \ref{solve_eq}. Hence, $\mathbb{P}_{XY|Z=1}$ has Lebesgue measure zero on the interval $[0,1]$, once we have defined all of the other variables in the equation. Thus, $\Sigma_{XY \cdot Z} = 0$ is not always equivalent to $X \ci Y | Z$, but satisfying the condition $\int \mathbb{P}_{XY|Z} ~ d\mathbb{P}_{Z} =  \int \mathbb{P}_{X|Z}{P}_{Y|Z} ~ d\mathbb{P}_{Z}$ when $X \not \ci Y | Z$ requires a very particular setup which is probably rarely encountered in practice.

The aforementioned argument motivates us to also consider the following statistic using a finite dimensional partial cross-covariance matrix:
\begin{equation}
\mathcal{S}^{\prime} = n\| \widehat{\Sigma}_{AB \cdot C} \|_F^2,
\end{equation}
where we have replaced $\ddot{A}$ with $A$. The above statistic is a generalization of linear partial correlation, because we consider uncorrelatedness of the residuals of non-linear functional transformations after performing non-linear regression. The asymptotic distribution for $\mathcal{S}$ in Theorem \ref{thm_asymp} also holds for $\mathcal{S}^{\prime}$, when we replace $\ddot{A}$ with $A$. Here, we use the hypotheses:
\begin{equation} \label{hypo3}
\begin{aligned}
&H_0: \| \Sigma_{AB \cdot C} \|_F^2 =0, \\
&H_1: \| \Sigma_{AB \cdot C} \|_F^2 > 0.
\end{aligned}
\end{equation}

In practice, the test which uses $\mathcal{S}^{\prime}$, which we now call the Randomized conditional Correlation Test (RCoT), usually rivals or outperforms RCIT and KCIT, because (1) nearly all conditionally dependent variables encountered in practice are also conditionally correlated after at least one functional transformation, and (2) we can easily calibrate the null distribution of the test using $\mathcal{S}^{\prime}$ even when $Z$ has large cardinality. We will therefore find this test useful for replacing RCIT when we have large conditioning set sizes ($\geq 4$).

\section{Experiments}

We carried out experiments to compare the empirical performance of the following tests:
\begin{itemize}
\item RCIT: uses $\mathcal{S}$ with the Lindsay-Pilla-Basak approximation,
\item RCoT: uses $\mathcal{S}^{\prime}$ with the Lindsay-Pilla-Basak approximation,
\item KCIT: uses $\mathcal{S}_K$ with a simulated null by bootstrap.
\end{itemize}
Note that KCIT with the gamma approximation performs \textit{slightly} faster than KCIT with bootstrap (e.g., less than 200ms faster on average at 2000 samples in our experiments), but the bootstrap results in a significantly better calibrated null distribution. We focus on large sample size ($\geq 500$) scenarios because we can just apply KCIT with bootstrap otherwise. We ran all experiments using the R programming language (Microsoft R Open) on a laptop with 2.60 GHz of CPU and 16GB of RAM. 

\subsection{Hyperparameters}
We used the same hyperparameters for RCIT and RCoT. Namely, we used the median Euclidean distance heuristic across the first 500 samples of $\ddot{X}$, $X$, $Y$ and $Z$ for choosing the $\sigma_{\ddot{X}},$ $\sigma_X,$ $\sigma_Y,$ and $\sigma_Z$ hyperparameters for the Gaussian RBF kernels $k_{\sigma}(x,y) = \text{exp}(- \| x-y\|^2 / \sigma)$, respectively\footnote{We also tried setting $\sigma_Z$ to the median distance divided by 1.5, 2 or 3. However, these values gave progressively worse performance on average.} \citep{Gretton08,LopezPaz14}. We also fixed the number of Fourier features for $\ddot{X}$, $X$ and $Y$ to 5 and the number of Fourier features for $Z$ to 25. We standardized all original and Fourier variables to mean zero unit variance in order to help ensure numerically stable computations. Finally, we set $\gamma$ to $1\text{E-}10$ in order to keep bias minimal. 

With KCIT, we set $\sigma$ to the squared median Euclidean distance between $(X,Y)$ using the first 500 samples times double the conditioning set size; the hyperparameters as described in the original paper, the hyperparameters in the author-provided MATLAB implementation and the hyperparameters of RCIT/RCoT all gave worse performance.

\subsection{Type I Error}
A good statistical test should control the Type I error rate at any specified $\alpha$. We therefore analyzed the Type I error rates of the three CI tests as a function of sample size and conditioning set size. We evaluated the algorithms using the Kolmogorov-Smirnov (KS) test statistic. Recall that the KS test uses the following statistic:
\begin{equation}
\mathcal{K} = \sup_{x \in \mathbb{R}} | \widehat{F}(x) - F(x) | = \| \widehat{F}_X - F_X \|_{\infty},
\end{equation}
\noindent where $\widehat{F}_X$ denotes the empirical CDF, and $F_X$ some comparison CDF. If the sample comes from $\mathbb{P}_X$, then $\mathcal{K}$ converges to 0 almost surely as $n \rightarrow \infty$ by the Glivenko-Cantelli theorem.

Now a good CI test controls the Type I error rate at any $\alpha$ value, when we have a uniform sampling distribution of the p-values over $[0,1]$. Therefore, a good CI test should have a small KS statistic value, when we set $F_X$ to the uniform distribution over $[0,1]$.

To compute the KS statistic values, we generated data from 1000 post non-linear models \citep{Zhang11,Doran14}. We can describe each post non-linear model as follows: $X=g_1(Z+\varepsilon_1), Y=g_2(Z+\varepsilon_2)$, where $Z,\varepsilon_1,\varepsilon_2$ have jointly independent standard Gaussian distributions, and $g_1,g_2$ denote smooth functions. We always chose $g_1,g_2$ uniformly from the following set of functions: $\{(\cdot), (\cdot)^2, (\cdot)^3, \text{tanh}(\cdot), \text{exp}(-\|\cdot\|_2 )\}$. Thus, we have $X \ci Y | Z$ in any case. Notice also that this situation is more general than the additive noise models proposed in \citep{Ramsey14}, where we have $X=g_1(Z)+\varepsilon_1, Y=g_2(Z)+\varepsilon_2$.

\subsubsection{Sample Size} \label{typeI_ss}
We first assess the Type I error rate as a function of sample size. We used sample sizes of 500, 1000, 2000, 5000, ten thousand, one hundred thousand and one million. A good CI test should control the Type I error rate across all $\alpha$ values at any sample size. Figure \ref{fig2:ss} summarizes the KS statistic values for the three different CI tests. Observe that all tests have similar KS statistic values across different sample sizes. We conclude that all three tests perform comparably in controlling the Type I error rate with a single conditioning variable at different sample sizes.

The run time results however tell a markedly different story. Both RCIT and RCoT output a p-value much more quickly than KCIT at different sample sizes (Figure \ref{fig2:rt}). Moreover, KCIT ran out of memory at 5000 samples while RCIT and RCoT handled one million samples in a little over 6 seconds. RCIT and RCoT also completed more than two orders of magnitude faster than KCIT on average at a sample size of 2000 (Figure \ref{fig2:rt_comp}). We conclude that RCIT and RCoT are more scalable than KCIT. Moreover, the experimental results agree with standard matrix complexity theory; RCIT and RCoT scale linearly with sample size, while KCIT scales cubicly with sample size.

\subsubsection{Conditioning Set Size} \label{typeI_dim}
CCD algorithms request p-values from CI tests using large conditioning set sizes. In fact, algorithms which do not assume causal sufficiency, such as FCI, often demand very large conditioning set sizes ($>5$). We should however also realize that CCD algorithms search for \textit{minimal} conditioning sets in order to establish ancestral relations. This means that we must focus on testing for cases where $X \not \ci Y |Z$, but we have either $X \ci Y |Z \cup A$ or $X \not \ci Y |Z \cup A$, where $|A|=1$.

We therefore evaluated the Type I error rates of the CI tests as a function of conditioning set size by fixing the sample size at 1000 and then adding 1 to 10 standard Gaussian variables into the conditioning set so that $X=g_1(\frac{1}{k} \sum_{j=1}^k Z_j+\epsilon_1), Y=g_2(\frac{1}{k} \sum_{j=1}^k Z_j+\epsilon_2), k=\{1,\dots,10\}$ in 1000 models. Note that this situation corresponds to 1 to 10 common causes.

Figure \ref{fig2:dim} summarizes the KS statistic values in the aforementioned scenario. We see that the KS statistic values for RCoT remain the smallest for nearly all conditioning set sizes, followed by RCIT and then KCIT. This implies that RCoT best approximates the null distribution out of the three CI tests. We also provide the histograms of the p-values across the 1000 post non-linear models at a conditioning set size of 10 for KCIT, RCIT, and RCoT in Figures \ref{fig2:histo_kcit}-\ref{fig2:histo_rcot}. Notice that the histograms become progressively more similar to a uniform distribution. We conclude that RCoT controls its Type I error rate the best even with large conditioning set sizes while KCIT controls its rate the worst. 

Now the run times of all three tests only increased very slightly with the conditioning set size (Figure \ref{fig2:dim_rt}). However, both RCIT and RCoT still completed 40.91 times faster than KCIT on average (95\% confidence interval: $\pm$0.44). These results agree with standard matrix complexity theory, as we expect all tests to scale linearly with dimensionality.

\begin{figure}
\begin{subfigure}{.5\textwidth}
  \centering
  \includegraphics[width=.8\linewidth]{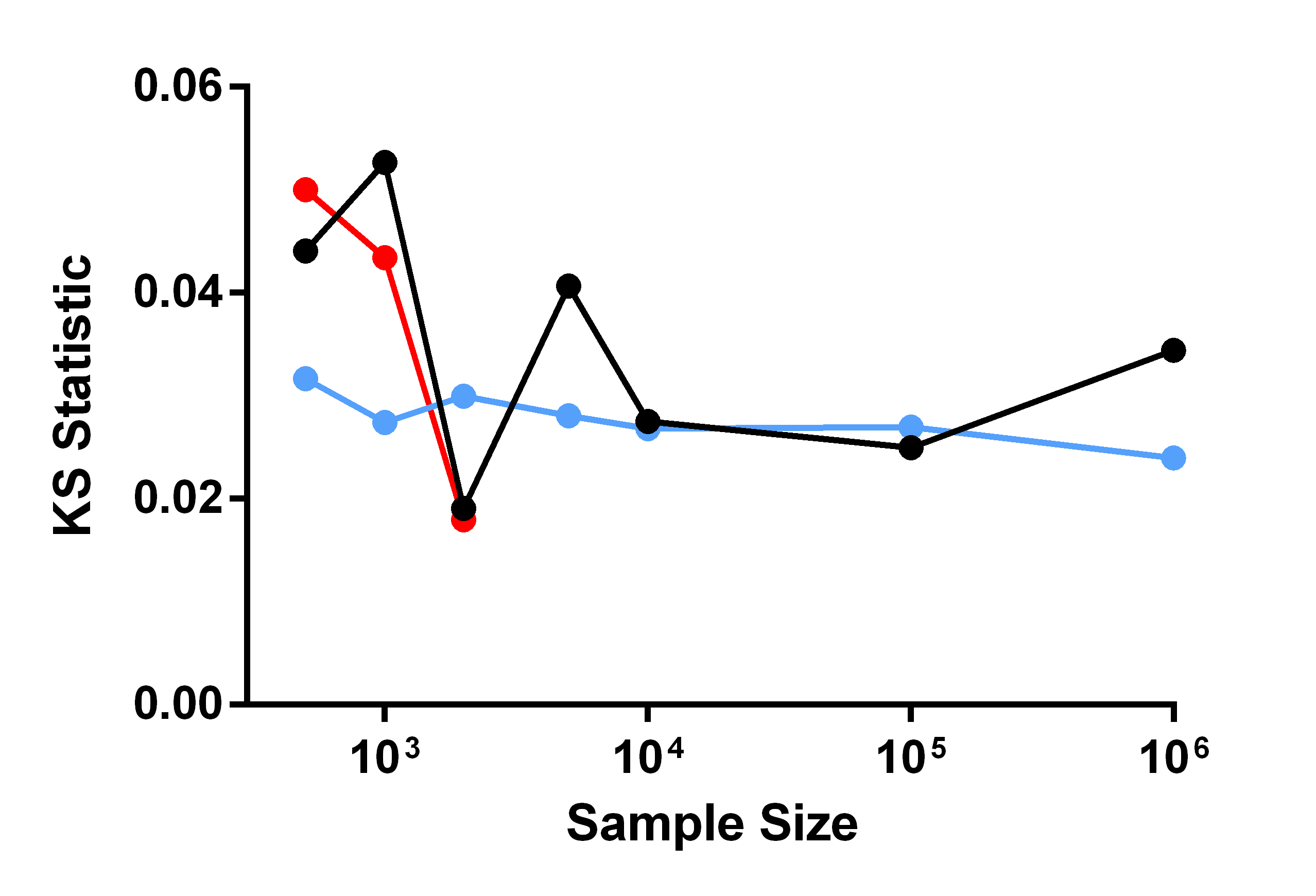}
  \caption{}
  \label{fig2:ss}
\end{subfigure}
\begin{subfigure}{.5\textwidth}
  \centering
  \includegraphics[width=.8\linewidth]{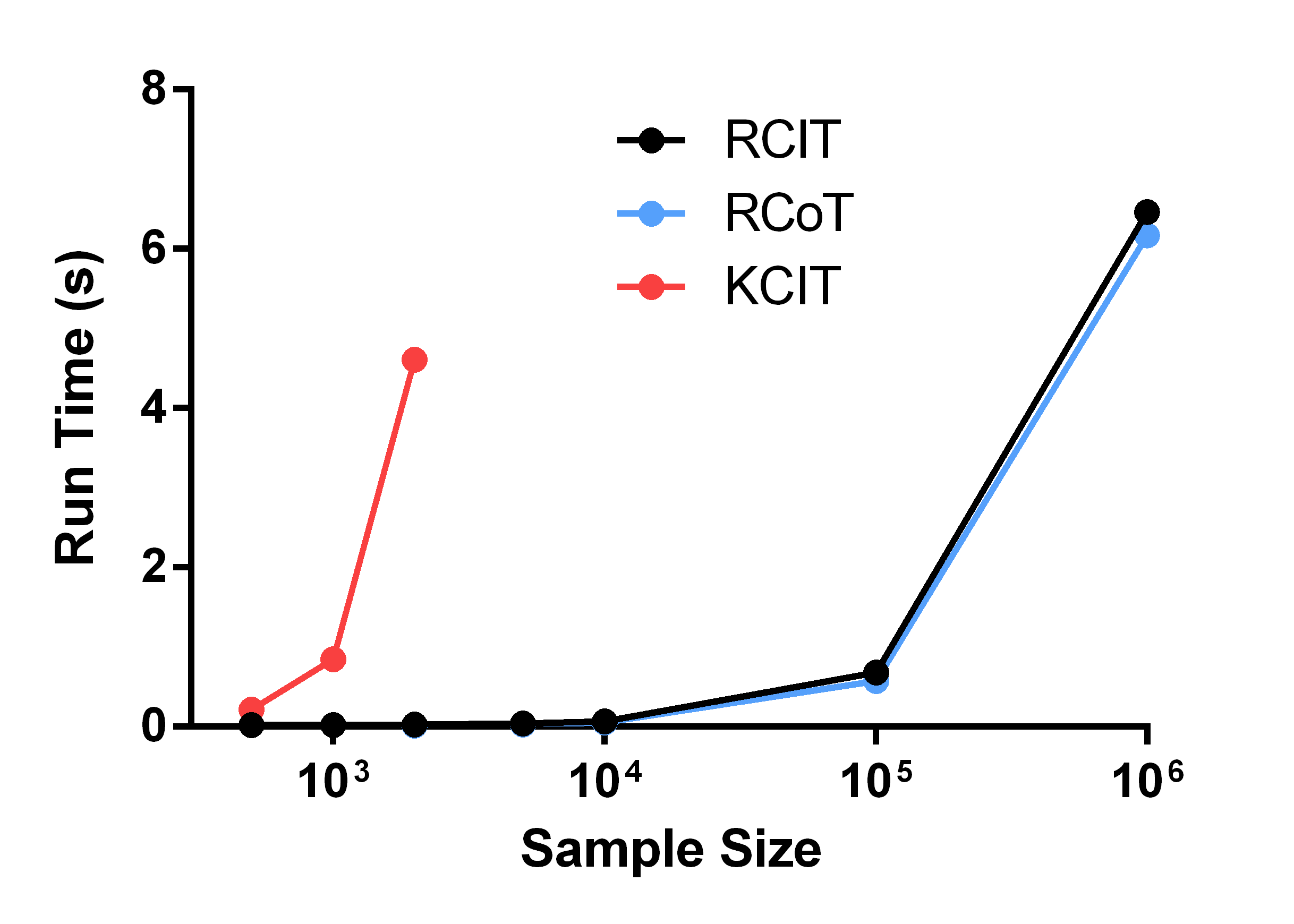}
  \caption{}
  \label{fig2:rt}
\end{subfigure}
\begin{subfigure}{.5\textwidth}
  \centering
  \includegraphics[width=.75\linewidth]{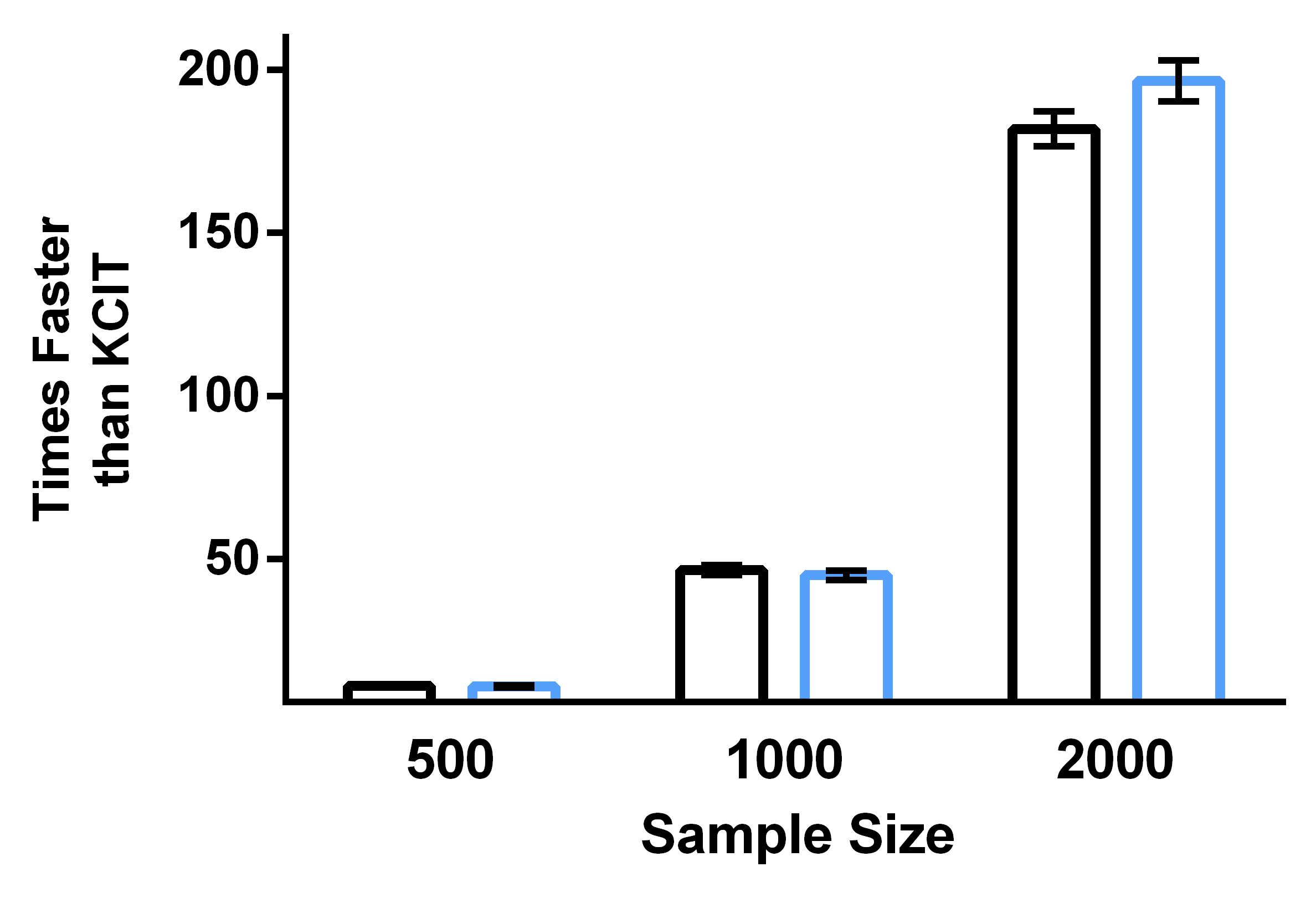}
  \caption{}
  \label{fig2:rt_comp}
\end{subfigure}
\begin{subfigure}{.5\textwidth}
  \centering
  \includegraphics[width=.8\linewidth]{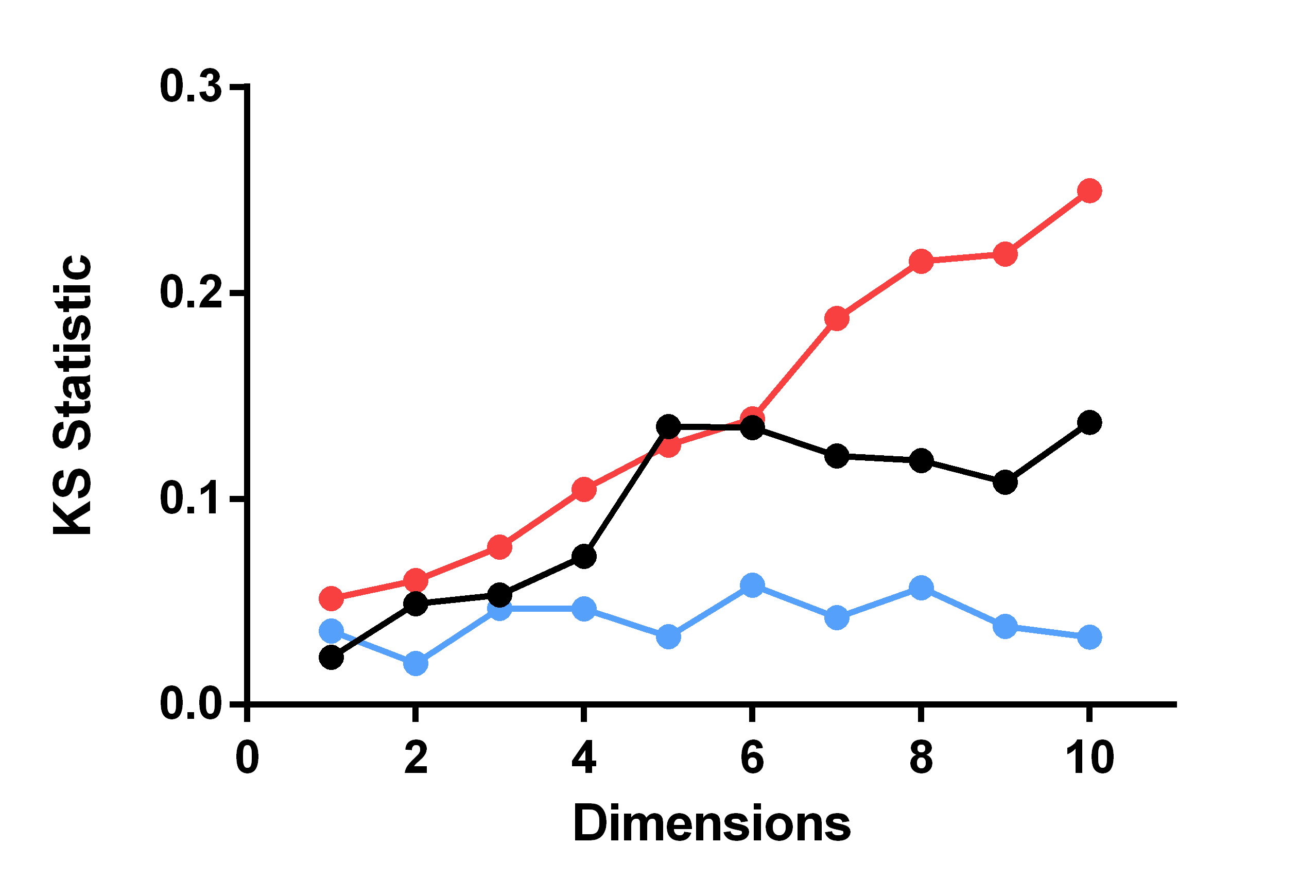}
  \caption{}
  \label{fig2:dim}
\end{subfigure}
\begin{subfigure}{.5\textwidth}
  \centering
  \includegraphics[width=.8\linewidth]{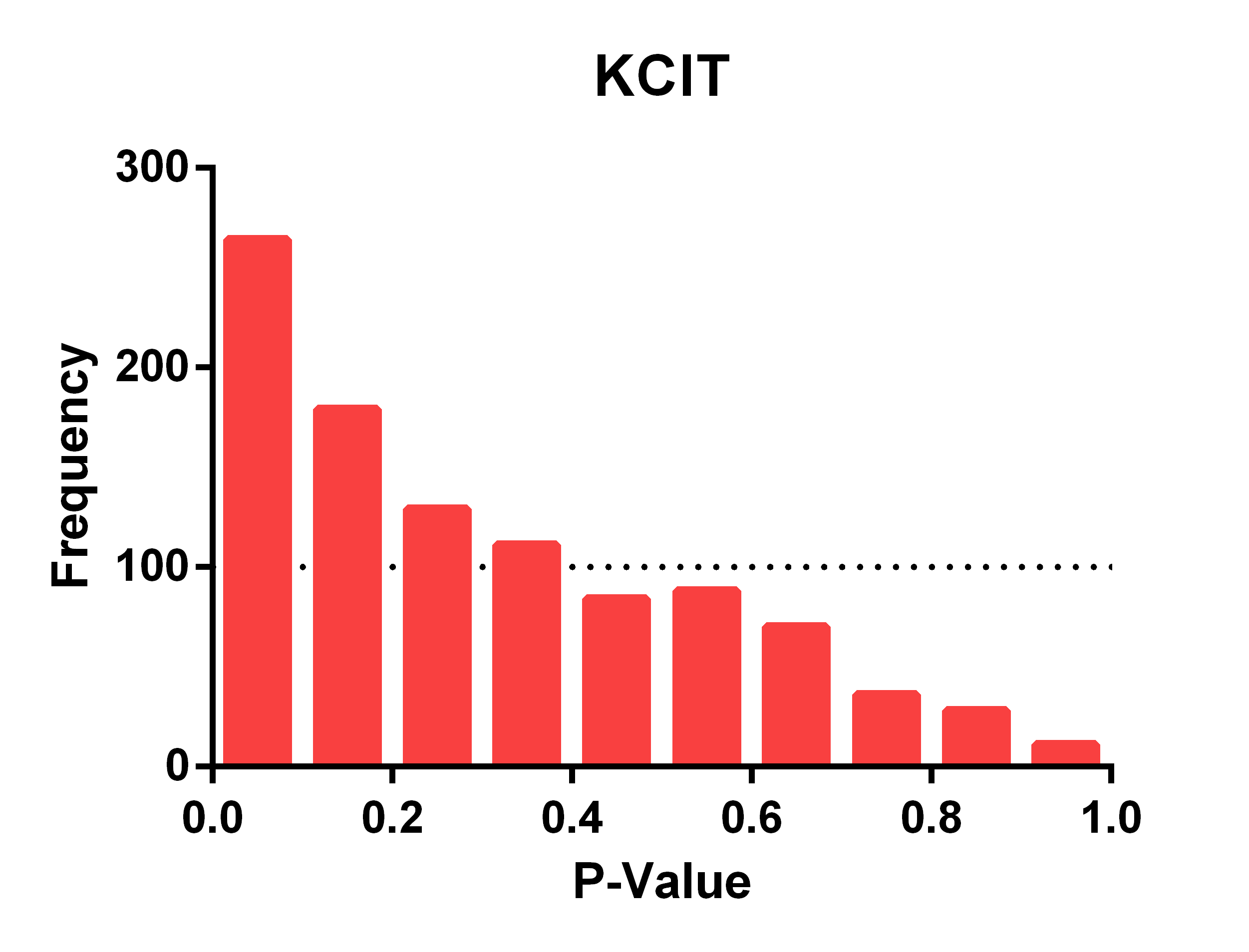}
  \caption{}
  \label{fig2:histo_kcit}
\end{subfigure}
\begin{subfigure}{.5\textwidth}
  \centering
  \includegraphics[width=.8\linewidth]{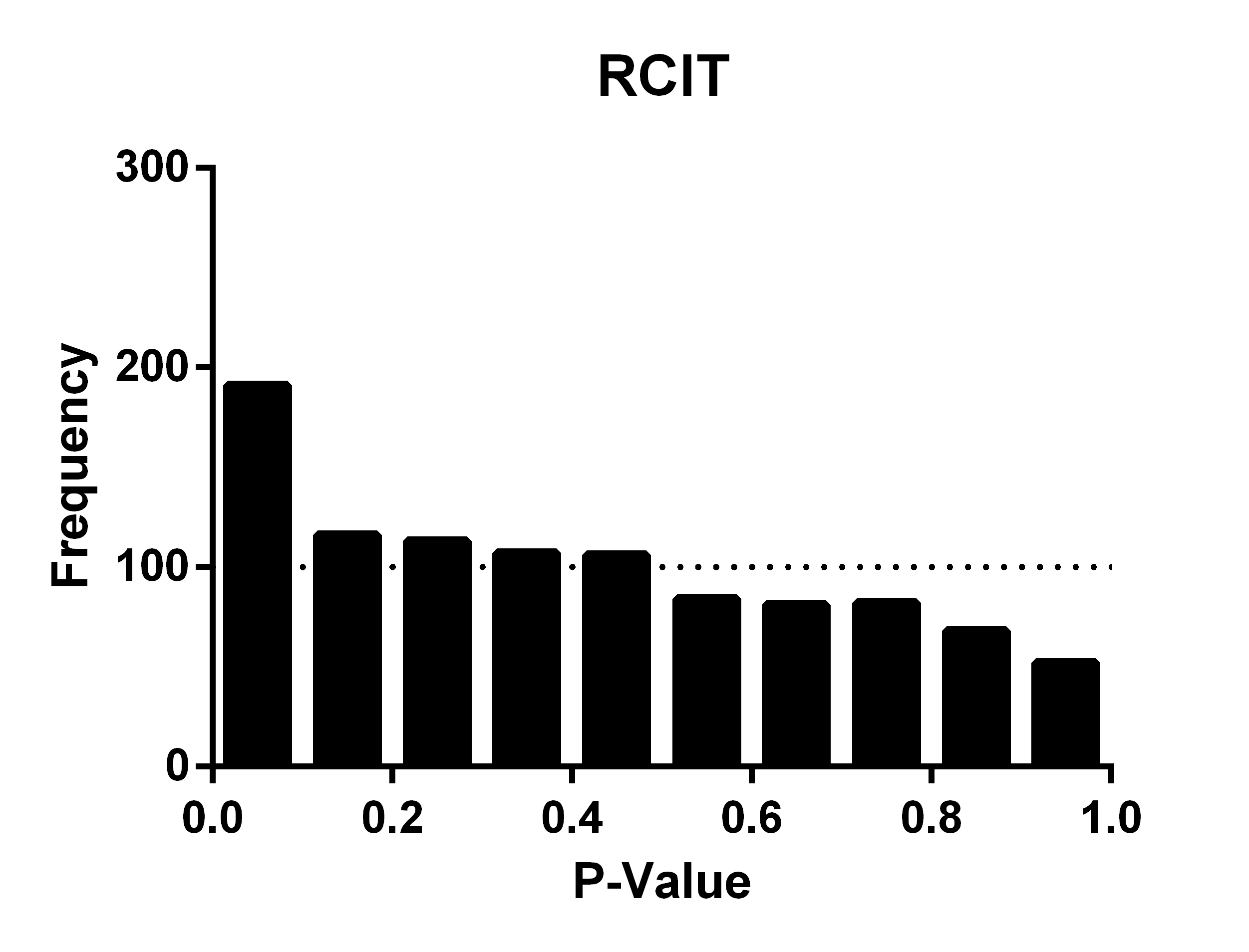}
  \caption{}
  \label{fig2:histo_rcit}
\end{subfigure}
\begin{subfigure}{.5\textwidth}
  \centering
  \includegraphics[width=.8\linewidth]{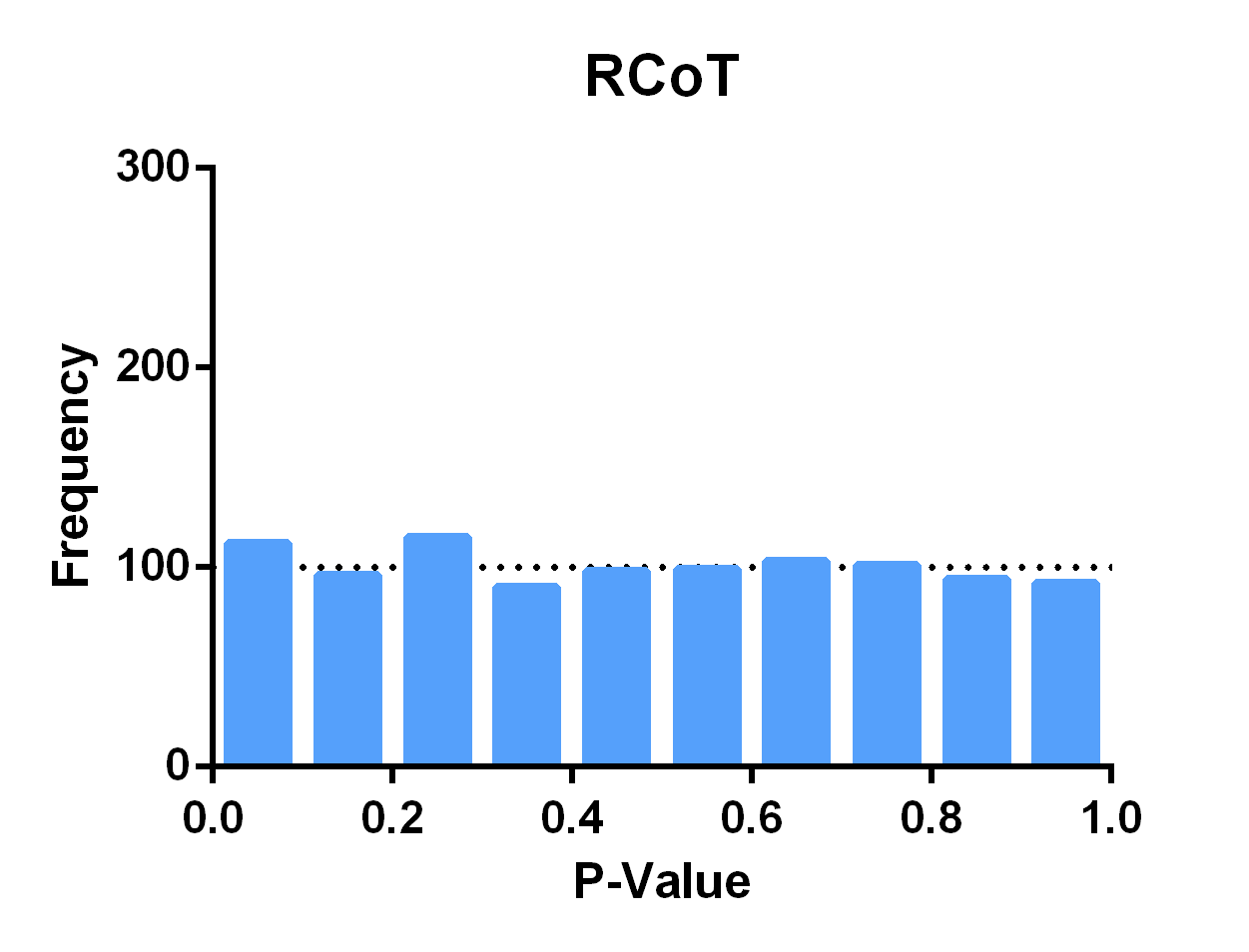}
  \caption{}
  \label{fig2:histo_rcot}
\end{subfigure}
\begin{subfigure}{.5\textwidth}
  \centering
  \includegraphics[width=.8\linewidth]{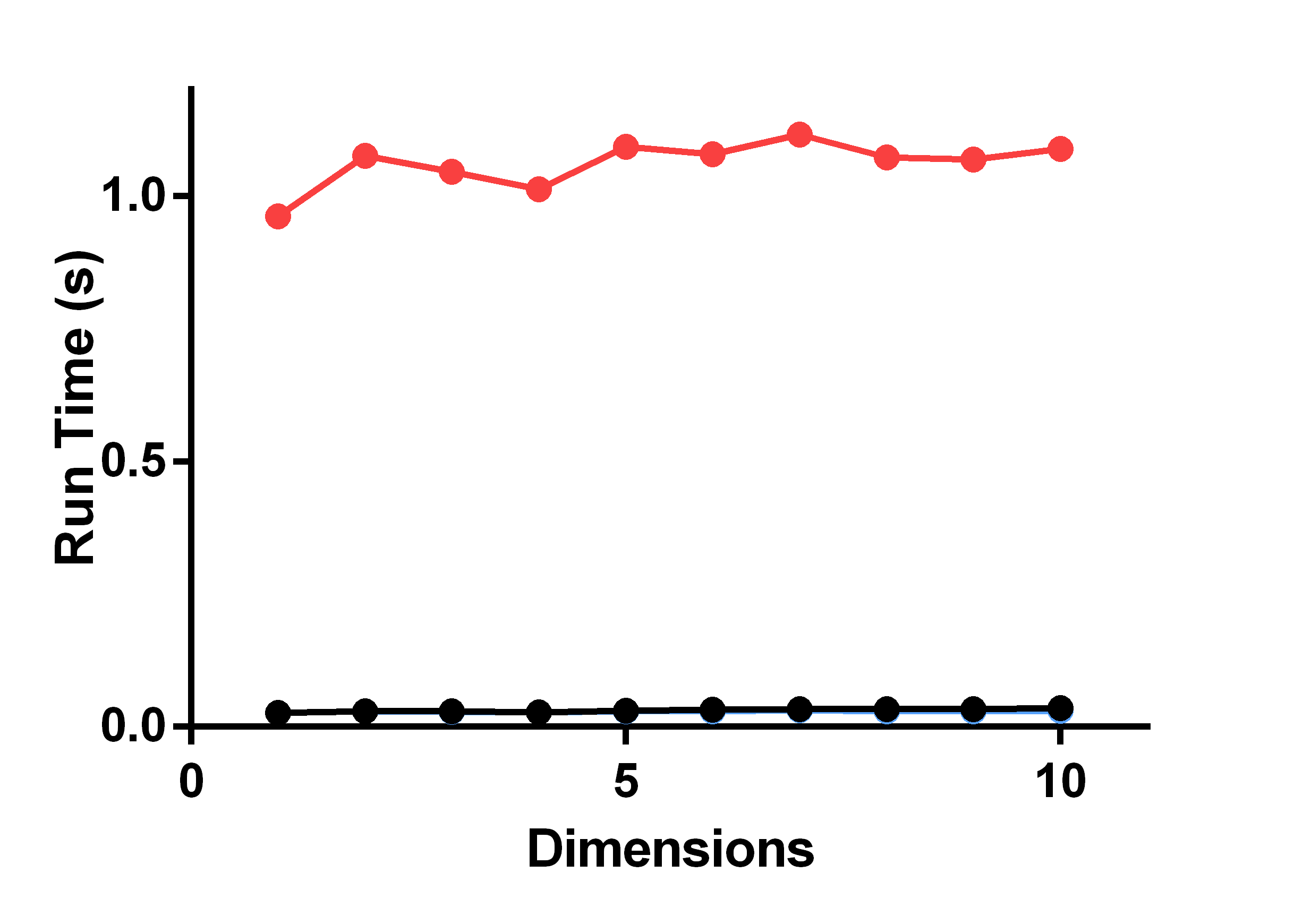}
  \caption{}
  \label{fig2:dim_rt}
\end{subfigure}

\caption{Experimental results of RCIT, RCoT and KCIT as a function of sample size and conditioning set size when conditional independence holds. (a) All tests have comparable KS statistic values as a function of sample size with a conditioning set size of one. (b) However, both RCIT and RCoT complete much faster than KCIT. (c) The relative difference in speed between RCIT vs. KCIT and RCoT vs. KCIT quickly grows with increasing sample size. (d) RCoT maintains the lowest KS statistic value with increases in the dimensionality of the conditioning set. (e-g) Histograms with a conditioning set size of 10 (the hardest case tested) show that KCIT, RCIT and RCoT obtain progressively more uniform null distributions. (h) Run times of all three tests scale linearly with dimensionality of the conditioning set. Both RCIT and RCoT have nearly identical run times in this case, so the black and blue lines overlap.} \label{fig2}
\end{figure}

\subsection{Power}
We next evaluated test power (i.e., $1-$(Type II error rate)) by computing the area under the power curve (AUPC). The AUPC corresponds to the area under the empirical CDF of the p-values returned by a CI test when the null does not hold. A CI test has higher power when its AUPC is closer to one. For example, observe that if a CI test always returns a p-value of 0 in the perfect case, then its AUPC corresponds to 1.

We examined the AUPC by adding the same small error $\varepsilon_b \sim \mathcal{N}(0,1/16)$ to both $X$ and $Y$ in 1000 post non-linear models as follows: $X=g_1(\varepsilon_b+\varepsilon_1), Y=g_2(\varepsilon_b+\varepsilon_2)$, $Z \sim \mathcal{N}(0,1)$. Here, we do not allow the CI tests to condition on $\varepsilon_b$, so we always have $X \not \ci Y | Z$; this situation therefore corresponds to a hidden common cause.

\subsubsection{Sample Size}
We first examine power as a function of sample size. We again tested sample sizes of 500, 1000, 2000, 5000, ten thousand, one hundred thousand, and one million. We have summarized the results in Figure \ref{fig3:ss}. Both RCIT and RCoT have comparable AUPC values to KCIT with sample sizes of 500, 1000 and 2000. At larger sample sizes, KCIT again did not scale due to insufficient memory, but the AUPC of both RCIT and RCoT continued to increase at similar values. We conclude that all three tests have similar power.

The run time results mimic those of Section \ref{typeI_ss}; RCIT and RCoT completed orders of magnitude faster than KCIT (Figures \ref{fig3:ss_rt} and \ref{fig3:ss_rt_comp}). 

\subsubsection{Conditioning Set Size}
We next examined power as a function of conditioning set size. To do this, we fixed the sample size at 1000 and set $Z=(Z_1,\dots,Z_k)$ with $Z \sim \mathcal{N}(0,I_k), k=\{1,\dots,10\}$ in the 1000 post non-linear models. We therefore examined how well the CI tests reject the null under an increasing conditioning set size with uninformative variables. A good CI test should either (1) maintain its power or, more realistically, (2) suffer a graceful decline in power with an increasing conditioning set size because none of the variables in the conditioning set are informative for rendering conditional independence by design.

We have summarized the results in Figure \ref{fig3:dim}. Notice that all tests have comparable AUPC values with small conditioning set sizes (between 1 and 3), but the AUPC value of KCIT gradually increases with increasing conditioning set sizes; the AUPC value should not increase under the current setup with a well-calibrated null because the extra variables are uninformative. To determine the cause of the unexpected increase in power, we permuted the values of $X$ in each run in order to assess the calibration of the null distribution. Figure \ref{fig3:dim_perm} summarizes the results, where we can see that only KCIT's KS statistic grows with an increasing conditioning set size. We can therefore claim that the increasing AUPC value of KCIT holds because of a badly calibrated null distribution with larger conditioning set sizes. We conclude that both RCIT and RCoT maintain steady power under an increasing conditioning set size while KCIT does not.

The run times in Figures \ref{fig3:dim_rt} and \ref{fig3:dim_perm_rt} again mimic those in Section \ref{typeI_dim} with RCIT and RCoT completing in a much shorter time frame than KCIT.

\begin{figure}
\begin{subfigure}{.5\textwidth}
  \centering
  \includegraphics[width=.8\linewidth]{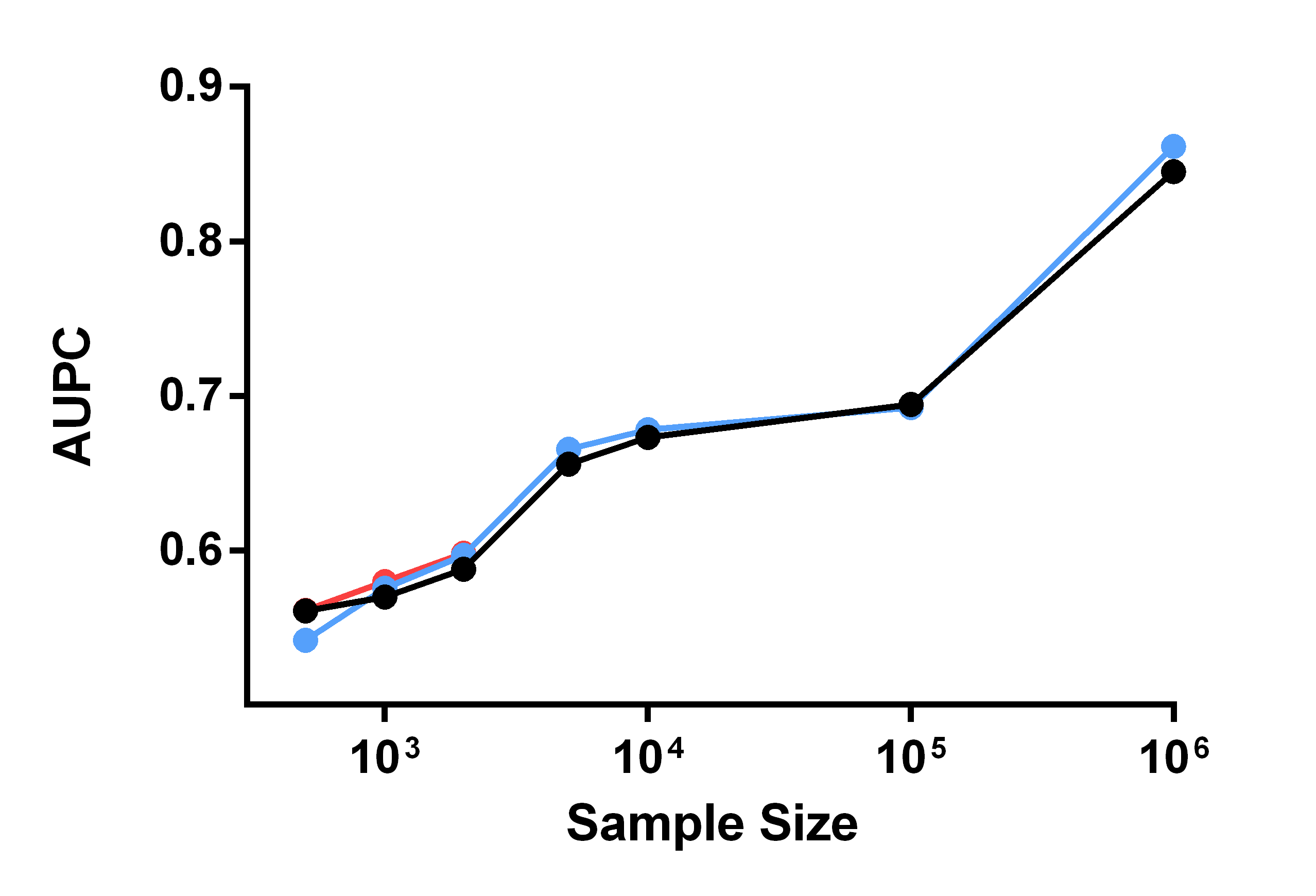}
  \caption{}
  \label{fig3:ss}
\end{subfigure}
\begin{subfigure}{.5\textwidth}
  \centering
  \includegraphics[width=.8\linewidth]{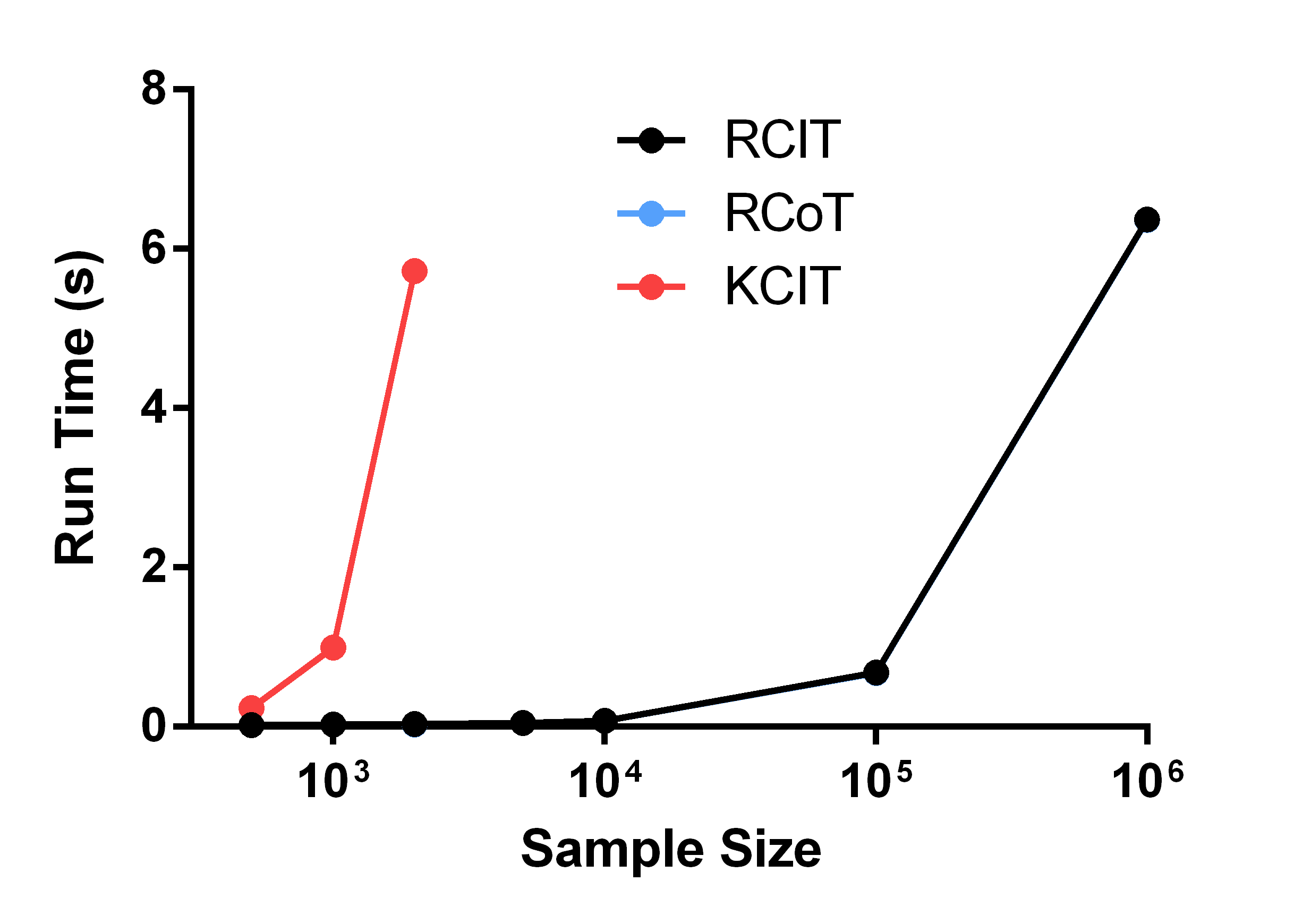}
  \caption{}
  \label{fig3:ss_rt}
\end{subfigure}
\begin{subfigure}{.5\textwidth}
  \centering
  \includegraphics[width=.8\linewidth]{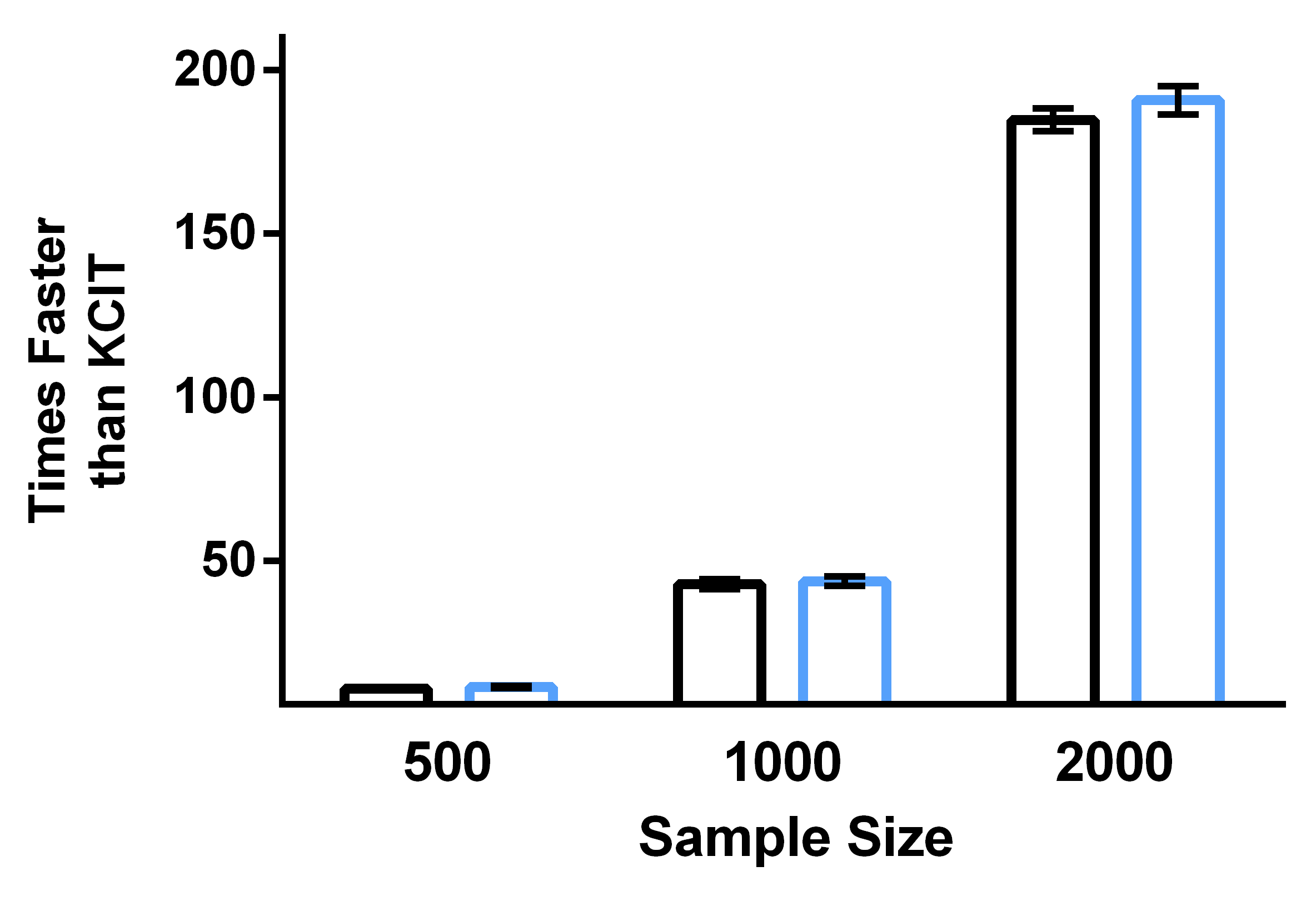}
  \caption{}
  \label{fig3:ss_rt_comp}
\end{subfigure}
\begin{subfigure}{.5\textwidth}
  \centering
  \includegraphics[width=.8\linewidth]{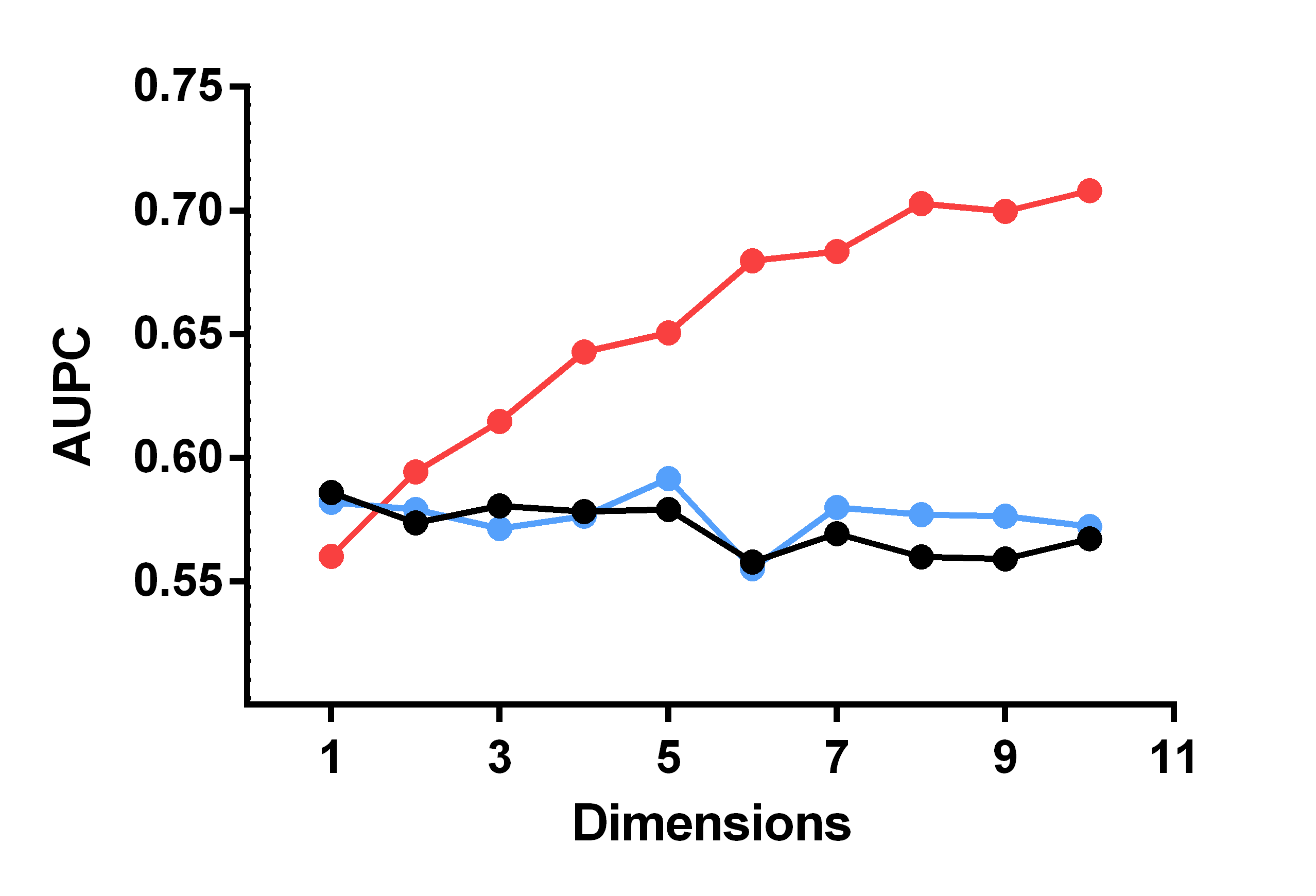}
  \caption{}
  \label{fig3:dim}
\end{subfigure}
\begin{subfigure}{.5\textwidth}
  \centering
  \includegraphics[width=.8\linewidth]{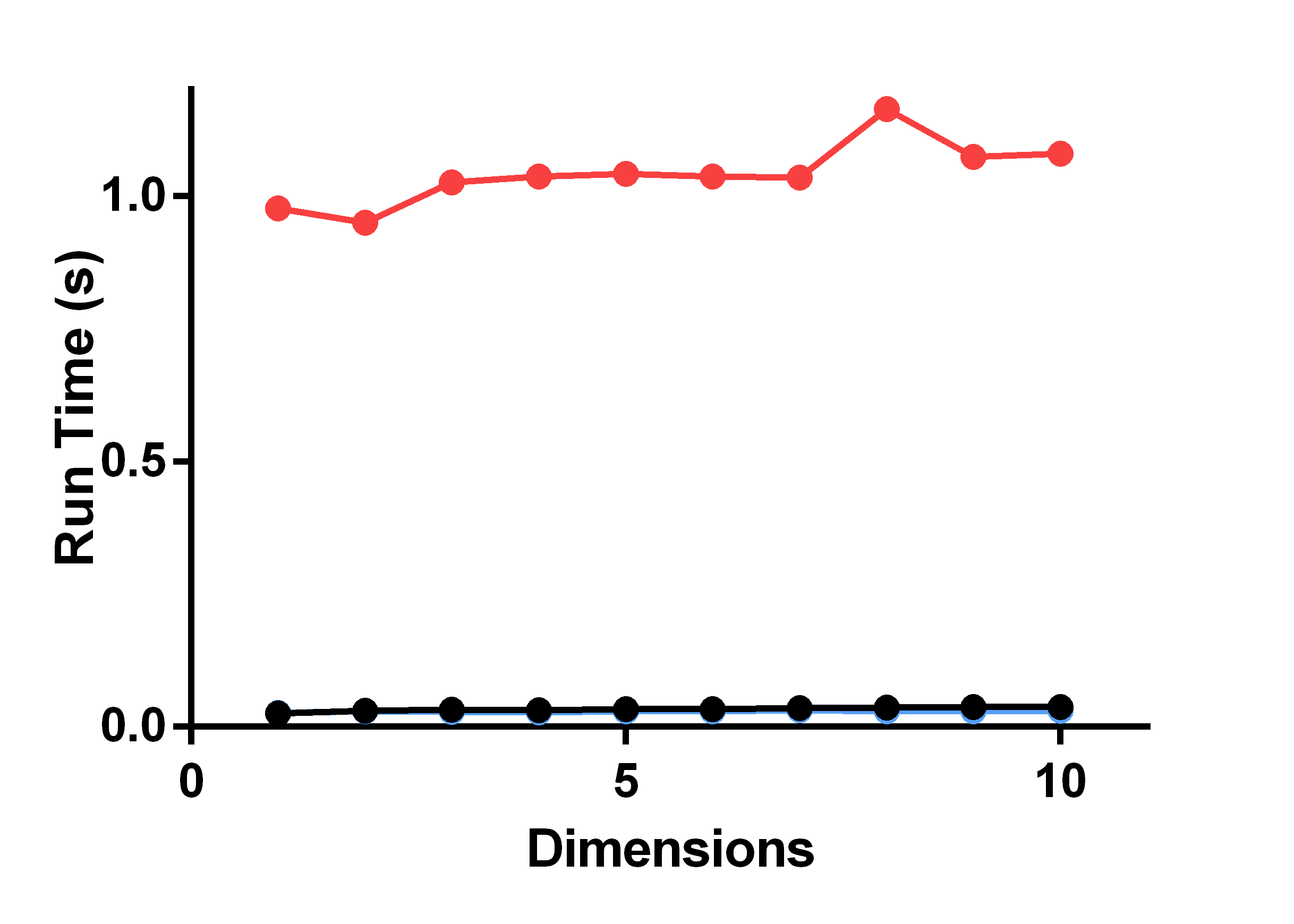}
  \caption{}
  \label{fig3:dim_rt}
\end{subfigure}
\begin{subfigure}{.5\textwidth}
  \centering
  \includegraphics[width=.8\linewidth]{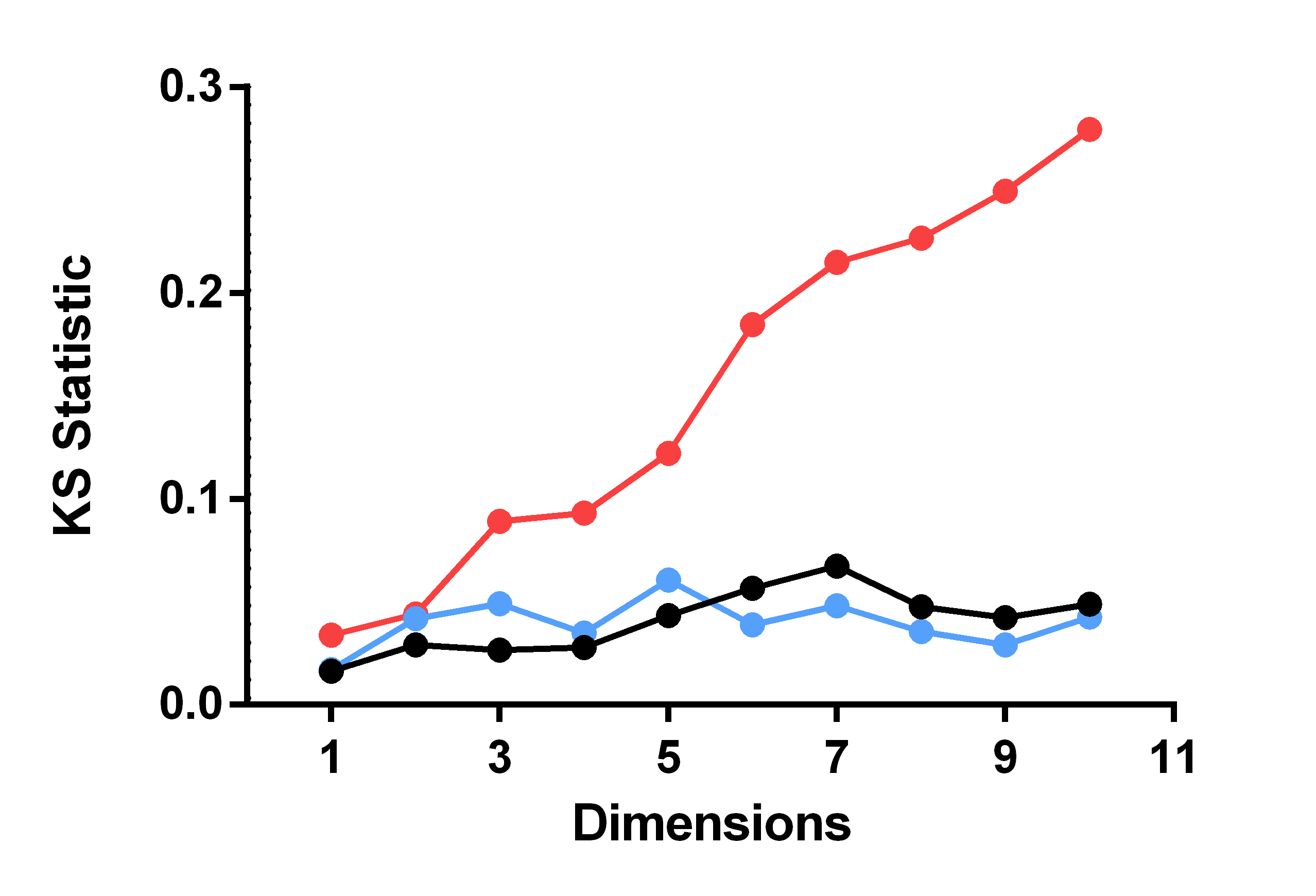}
  \caption{}
  \label{fig3:dim_perm}
\end{subfigure}
\begin{subfigure}{.5\textwidth}
  \centering
  \includegraphics[width=.8\linewidth]{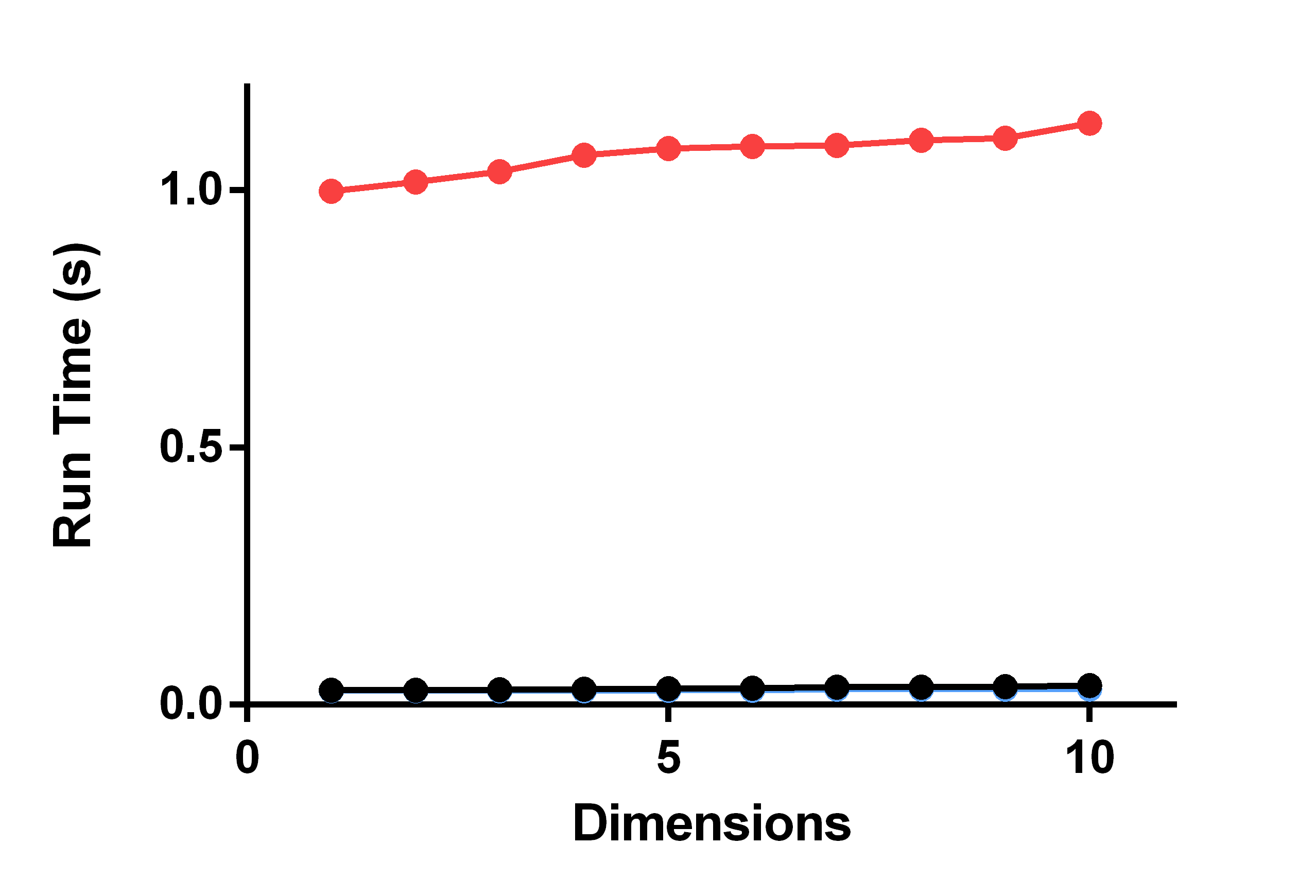}
  \caption{}
  \label{fig3:dim_perm_rt}
\end{subfigure}

\caption{Experimental results with RCIT, RCoT and KCIT as a function of sample size and conditioning set size when conditional dependence holds. (a) All tests have comparable AUPC values as a function of sample size with a conditioning set size of one. (b-c) Both RCIT and RCoT again complete much faster than KCIT. (d) KCIT's AUPC value unexpectedly increases with the dimensionality of the conditioning set. Associated run times for (d) in (e). (f) The cause of KCIT's AUPC increase lies in a badly calibrated null distribution; here we see that only KCIT's KS statistic value increases under the null. Associated run times for (f) in (g).} \label{fig3}
\end{figure}

\subsection{Causal Structure Discovery}
We next examine the accuracy of graphical structures as recovered by PC \citep{Spirtes00}, FCI \citep{Zhang08} and RFCI \citep{Colombo12} when run using RCIT, RCoT or KCIT.

We used the following procedure in \citep{Colombo12} to generate 250 different Gaussian DAGs with an expected neighborhood size $\mathbb{E}(N)=2$ and $v=20$ vertices. First, we generated a random adjacency matrix $\mathcal{A}$ with independent realizations of $\text{Bernoulli}(\mathbb{E}(N)/(v - 1))$ random variables in the lower triangle of the matrix and zeroes in the remaining entries. Next, we replaced the ones in $\mathcal{A}$ by independent realizations of a $\text{Uniform}([-1,-0.1]\cup[0.1, 1])$ random variable. We interpret a nonzero entry $\mathcal{A}_{ij}$ as an edge from $X_i$ to $X_j$ with
coefficient $\mathcal{A}_{ij}$ in the following linear model:
\begin{equation}
\begin{aligned}
&X_1 = \varepsilon_1, \\
&X_i = \sum_{r=1}^{v-1} \mathcal{A}_{ir}X_r + \varepsilon_i.
\end{aligned}
\end{equation}
for $i = 2, \dots , v$ where $\varepsilon_1, . . ., \varepsilon_v$ are mutually independent standard Gaussian random variables. The variables $\{X_1,\dots, X_v\} = \bm{X}$ then have a multivariate Gaussian distribution with mean $0$ and covariance matrix $\Sigma = (I_v - \mathcal{A})^{-1}(I_v - \mathcal{A})^{-T}$, where $I_v$ is the $v \times v$ identity matrix. To introduce non-linearities, we passed each variable in $\bm{X}$ through a non-linear function $g$ again chosen uniformly from the set $\{(\cdot)$, $(\cdot)^2$, $(\cdot)^3$, $\text{tanh}(\cdot)$, $\text{exp}(-\|\cdot\|_2 )\}$.

For FCI and RFCI, we introduced latent and selection variables using the following procedure. For each DAG, we first randomly selected a set of 0-3 latent common causes $L$. From the set $X \setminus L$, we then selected a set of 0-3 colliders as selection variables $S$. For each selection variable in $S$, we subsequently eliminated the bottom $q$ percentile of samples, where we drew $q$ according to independent realizations of a $\text{Uniform}([0.1,0.5])$ random variable. We finally eliminated all of the latent variables from the dataset. 

We ultimately created 250 different 500 sample datasets for PC, FCI and RFCI. We then ran the sample versions of PC, FCI and RFCI using RCIT, RCoT, KCIT and Fisher's z-test (FZT) at $\alpha=0.05$. We also obtained the oracle graphs by running the oracle versions of PC, FCI and RFCI using the ground truth. 

We have summarized the results as structural Hamming distances (SHDs) from the oracle graphs in Figure \ref{fig4:shd}. PC run with RCIT and PC run with RCoT both outperformed PC run with KCIT by a large margin according to paired t-tests (PC RCIT vs. KCIT, $t=\text{-}14.76, p < 2.2\text{E-}16$; PC RCoT vs. KCIT, $t=\text{-}12.87, p < 2.2\text{E-}16$). We found similar results with FCI and RFCI, although by only a small margin; 3 of the 4 comparisons fell below the Bonferonni corrected threshold of 0.05/6 and the other comparison fell below the uncorrected threshold of 0.05 (FCI RCIT vs. KCIT, $t=\text{-}2.00, p = 0.047$; FCI RCoT vs. KCIT $t=\text{-}2.96, p = 0.0034$; RFCI RCIT vs. KCIT, $t=\text{-}3.56, p = 4.5\text{E-}4$; RFCI RCoT vs. KCIT, $t=\text{-}2.80, p = 0.0055$). All algorithms with any of the kernel-based tests outperformed the same algorithms with FZT by a large margin ($p<7\text{E-}14$ in all cases). Finally, the run time results in Figure \ref{fig4:rt} show that the CCD algorithms run with RCIT and RCoT complete at least 13 times faster on average than those run with KCIT. We conclude that both RCIT and RCoT help CCD algorithms at least match the performance of the same algorithms run with KCIT, but RCIT and RCoT do so within a much shorter time frame than KCIT.

\begin{figure}
\begin{subfigure}{.5\textwidth}
  \centering
  \includegraphics[width=1\linewidth]{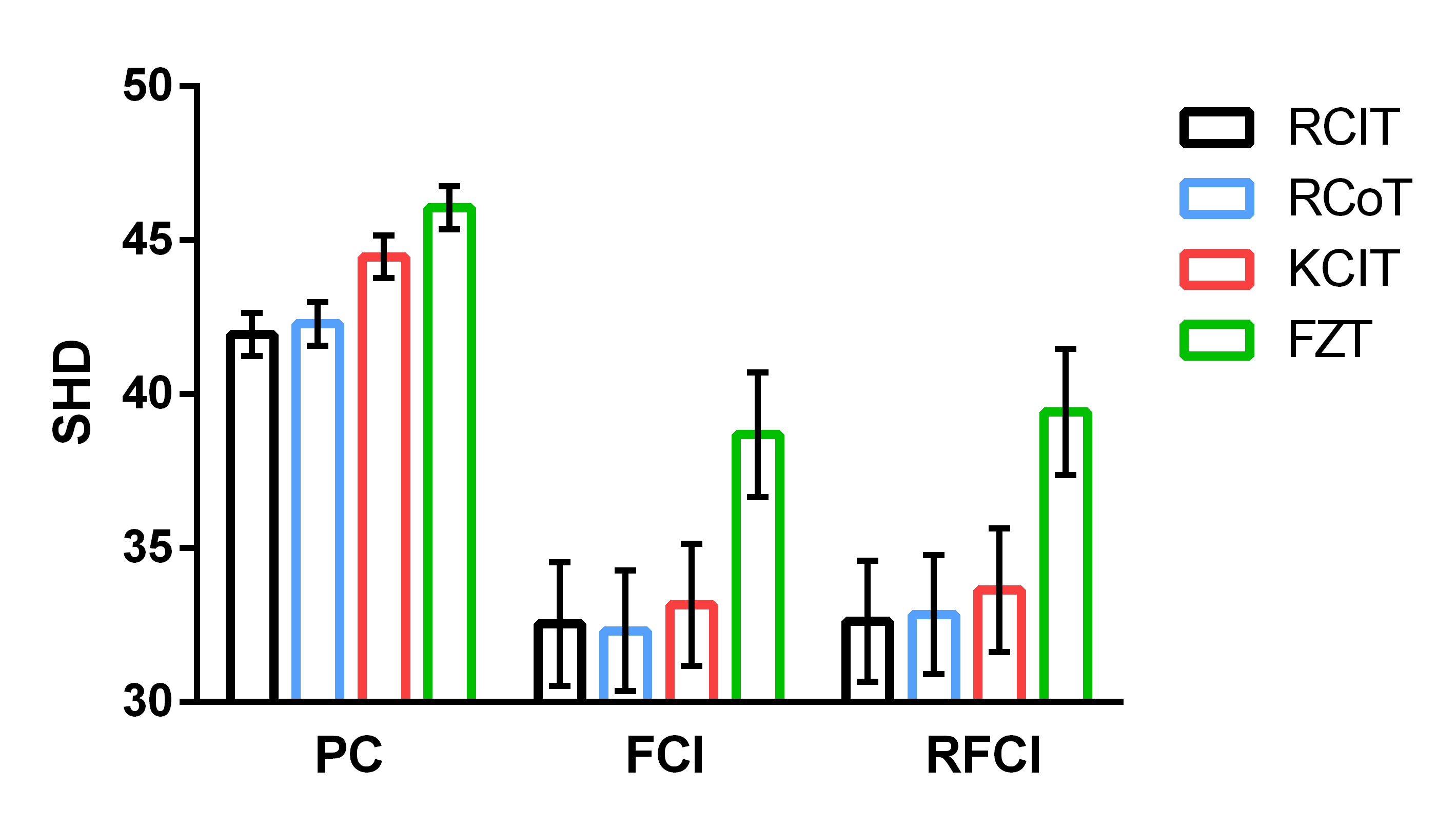}
  \caption{}
  \label{fig4:shd}
\end{subfigure}
\begin{subfigure}{.5\textwidth}
  \centering
  \includegraphics[width=.7\linewidth]{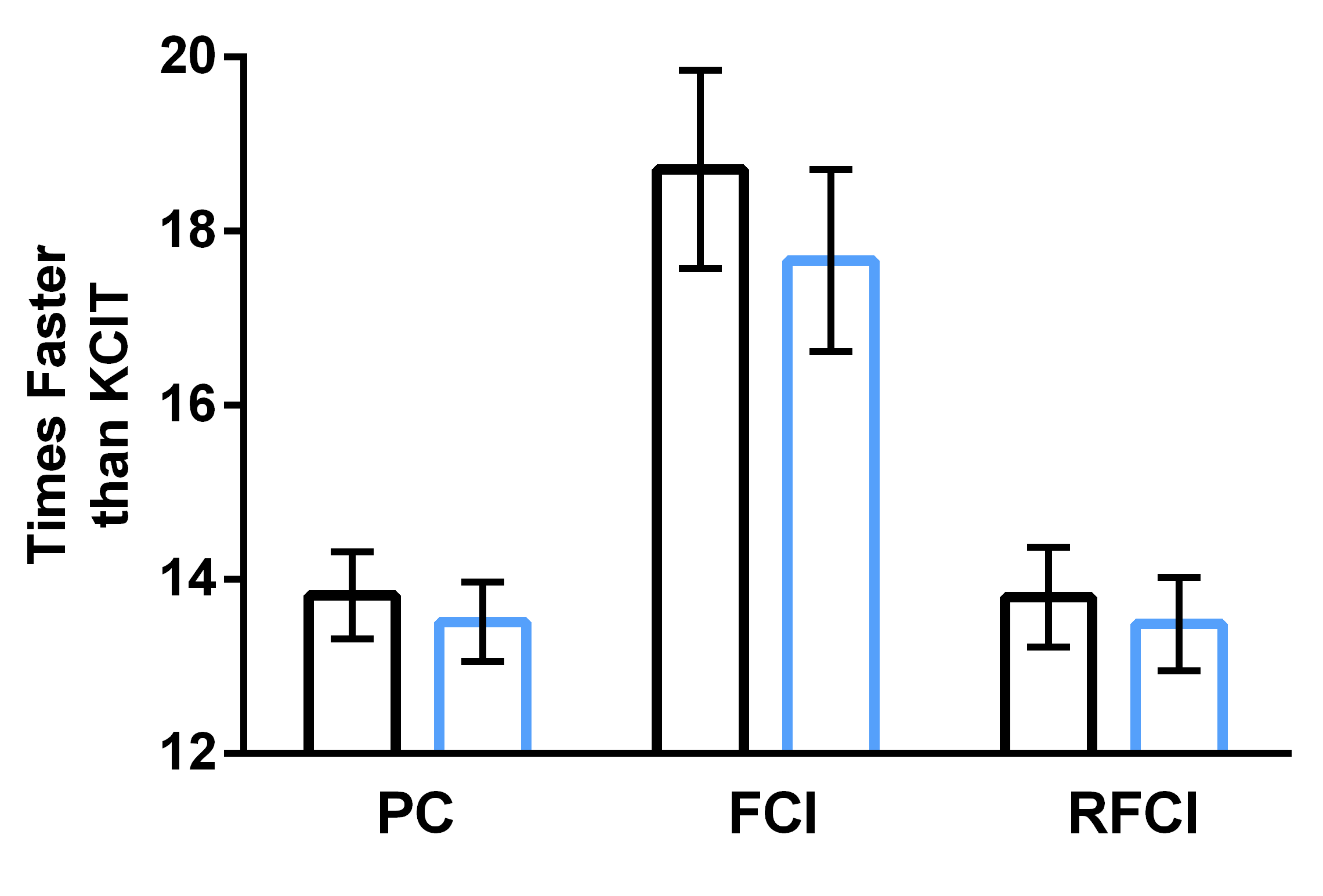}
  \caption{}
  \label{fig4:rt}
\end{subfigure}

\caption{Results of CCD algorithms as evaluated by mean (a) SHD and (b) run times. The CCD algorithms run with KCIT perform comparably (or even slightly worse) to those run with RCIT and RCoT in (a). Run times in (b) show that the CCD algorithms run with RCIT and RCoT complete at least 13 times faster on average than those with KCIT. Error bars denote 95\% confidence intervals of the mean.}
\end{figure}

\subsection{Real Data}

We finally ran PC, FCI and RFCI using RCIT, RCoT, KCIT and FZT at $\alpha=0.05$ on a publicly available longitudinal dataset from the Cognition and Aging USA (CogUSA) study \citep{McArdle15}, where scientists measured the cognition of men and women above 50 years of age. The dataset contains 815 samples, 18 variables and two waves (thus $18/2 = 9$ variables in each wave) separated by two years after some data cleaning\footnote{We specifically removed redundant variables with deterministic relations, variables with more than 1000 missing values, and then samples with missing values in any of the remaining variables.}. Note that we do not have access to a gold standard solution set in this case. However, we can utilize the time information in the dataset to detect false positive ancestral relations directed backwards in time.

We ran the CCD algorithms on 30 bootstrapped datasets. Results are summarized in Figure \ref{fig5}. Comparisons with PC did not reach the Bonferonni level among the kernel-based tests, although PC run with either RCIT or RCoT yielded fewer false positive ancestral relations on average than PC run with KCIT near an uncorrected level of 0.05 (PC RCIT vs. KCIT, $t=\text{-}2.76, p = 9.85\text{E-}3$; PC RCoT vs. KCIT, $t=\text{-}1.99, p = 0.056$). However, FCI and RFCI run with either RCIT or RCoT performed better than those run with KCIT at a Bonferroni corrected level of 0.05/6 (FCI RCIT vs. KCIT, $t=\text{-}29.57, p < 2.2\text{E-}16$; FCI RCoT vs. KCIT, $t=\text{-}17.41, p < 2.2\text{E-}16$; RFCI RCIT vs. KCIT, $t=\text{-}6.50, p = 4.13\text{E-}7$; RFCI RCoT vs. KCIT, $t=\text{-}7.39, p = 3.85\text{E-}8$). The CCD algorithms run with FZT also gave inconsistent results; PC run with FZT performed the best on average, but FCI and RFCI run with FZT also performed second from the worst. Here, we should trust the outputs of FCI and RFCI more strongly than those of PC, since both FCI and RFCI allow latent common causes and selection bias which often exist in real data. Next, CCD algorithms run with RCIT performed comparably to those run with RCoT (PC RCIT vs. RCoT, $t=\text{-}1.05, p = 0.301$; FCI RCIT vs. RCoT, $t=\text{-}1.54, p = 0.134$; RFCI RCIT vs. RCoT, $t=\text{-}0.89, p = 0.380$). We finally report that the CCD algorithms run with RCIT and RCoT complete at least 40 times faster on average than those run with KCIT (Figure \ref{fig5:run_time}). We conclude that CCD algorithms run with either RCIT or RCoT perform at least as well as those run with KCIT on this real dataset but with large reductions run time.

\begin{figure}
\begin{subfigure}{0.5\textwidth}
  \centering
  \includegraphics[width=0.8\linewidth]{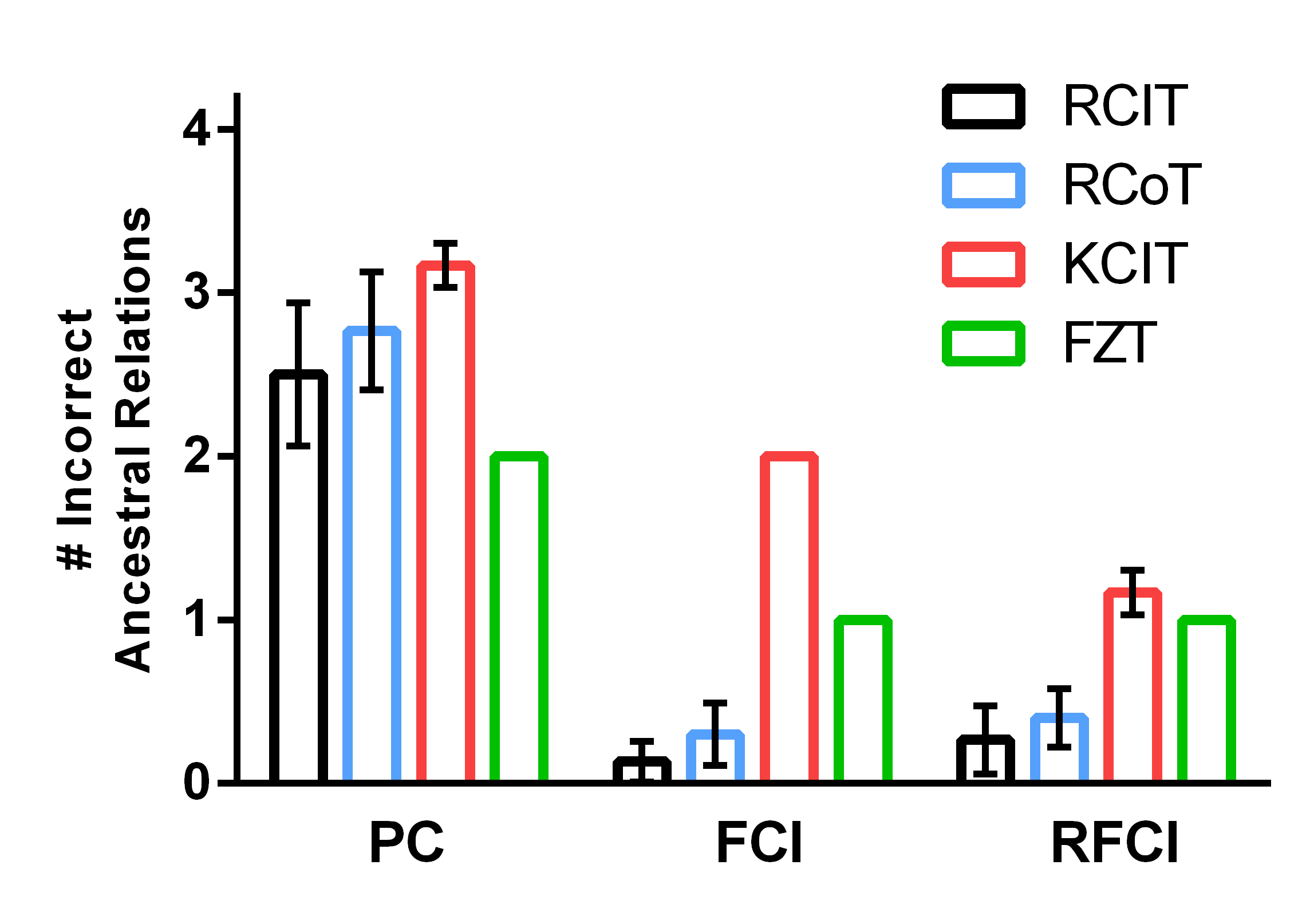}
  \caption{}
  \label{fig5:ancestral}
\end{subfigure}
\begin{subfigure}{0.5\textwidth}
  \centering
  \includegraphics[width=0.8\linewidth]{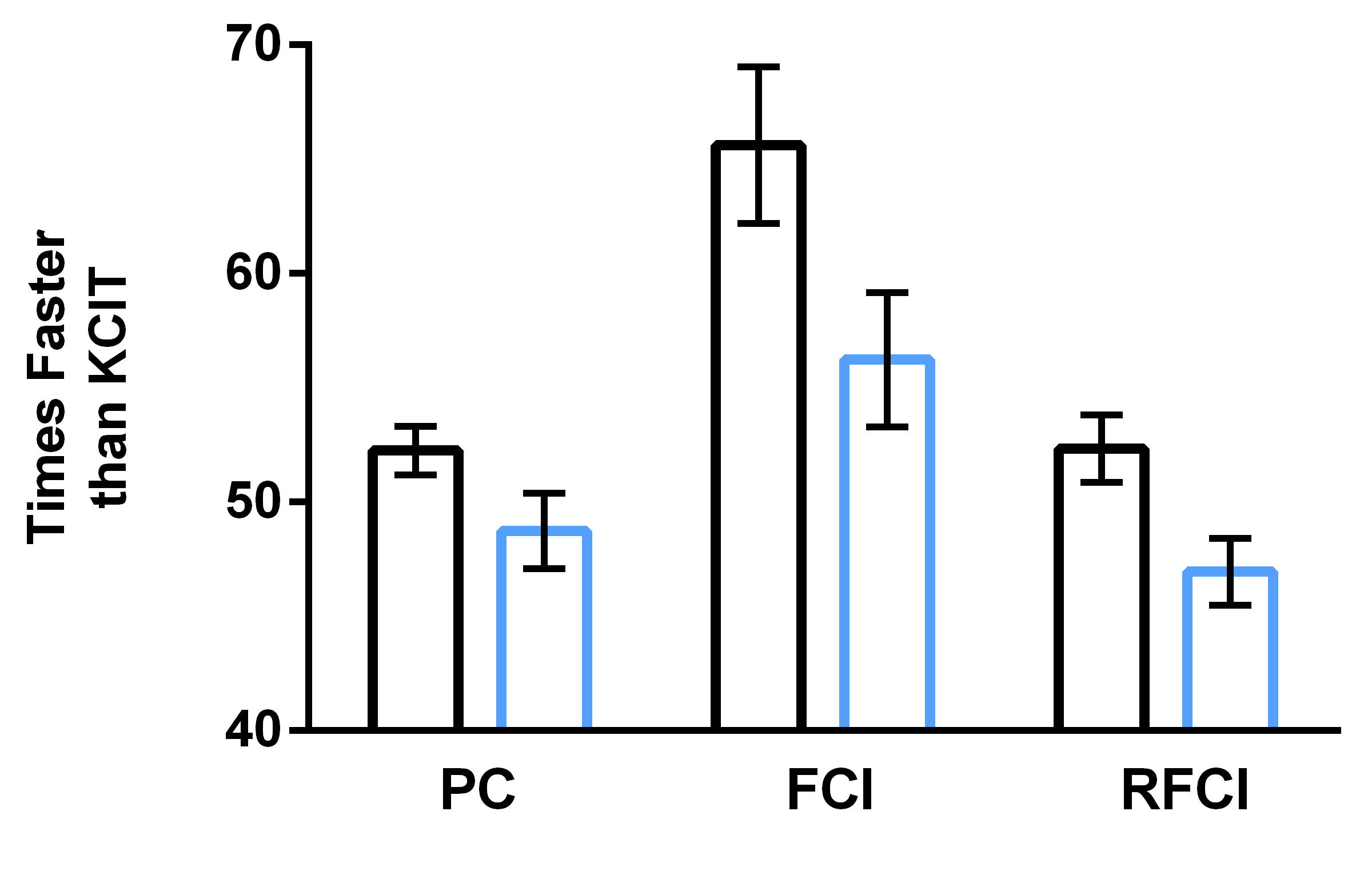}
  \caption{}
  \label{fig5:run_time}
\end{subfigure}

\caption{Results of CCD algorithms as evaluated on real longitudinal data. Part (a) displays mean counts of the number of ancestral relations directed backwards time. We do not display 95\% confidence intervals when we computed a standard error of zero. Part (b) summarizes the mean run times.} \label{fig5}
\end{figure}

\section{Conclusion}

We developed two statistical tests called RCIT and RCoT for fast non-parametric CI testing. Both RCIT and RCoT approximate KCIT by sampling Fourier features. Moreover, the proposed tests return p-values orders of magnitude faster than KCIT in the large sample size setting. RCoT in particular also has a better calibrated null distribution than KCIT especially with larger conditioning set sizes. In causal graph discovery, RCIT and RCoT help CCD algorithms recover graphical structures at least as accurately as KCIT but, most importantly, also allow the algorithms to complete in a much shorter time frame. We believe that the speedups provided by RCIT and RCoT will make non-parametric causal discovery more accessible to scientists who wish to apply CCD algorithms to their datasets.

\section*{Acknowledgments}
Research reported in this publication was supported by grant U54HG008540 awarded by the National Human Genome Research Institute through funds provided by the trans-NIH Big Data to Knowledge initiative. The research was also supported by the National Library of Medicine of the National Institutes of Health under award numbers T15LM007059 and R01LM012095. The content is solely the responsibility of the authors and does not necessarily represent the official views of the National Institutes of Health.

\section*{References}
\bibliographystyle{abbrvnat}
\bibliography{thesis_biblio}

\section{Appendix}

We will prove the central limit theorem (CLT) for the sample covariance matrix. We first have the following sample covariance matrices with known and unknown expectation vector, respectively:
\begin{equation}
\begin{aligned}
\ddot{\Sigma} &= \frac{1}{n} \sum_{i=1}^n \big[X_i-\mathbb{E}(X) \big]\big[X_i-\mathbb{E}(X) \big]^T,\\
\widehat{\Sigma} &= \frac{1}{n-1} \sum_{i=1}^n \big[X_i-\widehat{\mathbb{E}}(X) \big]\big[X_i-\widehat{\mathbb{E}}(X) \big]^T. \\
\end{aligned}
\end{equation}

Now observe that we may write:
\begin{equation} \label{rewrite1}
\begin{aligned}
&\hspace{4.5mm} (n-1)\widehat{\Sigma}\\
& = \sum_{i=1}^n \big[X_i-\mathbb{E}(X)-(\widehat{\mathbb{E}}(X) - \mathbb{E}(X)) \big] \big[X_i-\mathbb{E}(X)-(\widehat{\mathbb{E}}(X) - \mathbb{E}(X)) \big]^T\\
&=  \sum_{i=1}^n (X_i-\mathbb{E}(X))(X_i-\mathbb{E}(X))^T + n(\widehat{\mathbb{E}}(X)-\mathbb{E}(X))(\widehat{\mathbb{E}}(X)-\mathbb{E}(X))^T\\
& \hspace{20mm} - 2(\widehat{\mathbb{E}}(X)-\mathbb{E}(X))\sum_{i=1}^n(X_i-\mathbb{E}(X))^T \\
& = n \ddot{\Sigma} - n(\widehat{\mathbb{E}}(X)-\mathbb{E}(X))(\widehat{\mathbb{E}}(X)-\mathbb{E}(X))^T
\end{aligned}  
\end{equation}
It follows that:
\begin{equation} \label{rewrite2}
\begin{aligned}
&\hspace{4.5mm} \sqrt{n}(\widehat{\Sigma} - \Sigma) \\
& = \sqrt{n} \big( \frac{n-1}{n-1} \widehat{\Sigma} - \Sigma \big) \\
& = \sqrt{n} \big( \frac{n}{n-1} \ddot{\Sigma} - \frac{n}{n-1} (\widehat{\mathbb{E}}(X)-\mathbb{E}(X))(\widehat{\mathbb{E}}(X)-\mathbb{E}(X))^T - \Sigma \big) \\
& = \frac{n\sqrt{n}}{n-1}\ddot{\Sigma} - \frac{n\sqrt{n}}{n-1} (\widehat{\mathbb{E}}(X)-\mathbb{E}(X))(\widehat{\mathbb{E}}(X)-\mathbb{E}(X))^T - \sqrt{n} \Sigma \\
& = \frac{n\sqrt{n}}{n-1} \ddot{\Sigma} - \frac{n\sqrt{n}}{n-1} (\widehat{\mathbb{E}}(X)-\mathbb{E}(X))(\widehat{\mathbb{E}}(X)-\mathbb{E}(X))^T - \frac{n-1}{n-1}\sqrt{n} \Sigma \\
& = \frac{n\sqrt{n}}{n-1} (\ddot{\Sigma} - \Sigma) - \frac{n\sqrt{n}}{n-1} (\widehat{\mathbb{E}}(X)-\mathbb{E}(X))(\widehat{\mathbb{E}}(X)-\mathbb{E}(X))^T + \frac{\sqrt{n}}{n-1} \Sigma
\end{aligned}  
\end{equation}

We are now ready to state the result:
\begin{lemma} \label{cov_CLT}
Let $X_1, \dots, X_n$ refer to a sequence of i.i.d. random k-vectors. Denote the expectation vector and covariance matrix of $X_1$ as $\mu_1$ and $\Sigma_1$, respectively. Assume that $\breve{\Sigma}_1=\textup{Cov}\big[ v_u((X_1-\mu_1)(X_1-\mu_1)^T)\big]$ is positive definite, where $v_u(M)$ denotes the vectorization of the upper triangular portion of a real symmetric matrix $M$. Then, we have:
\begin{equation}
\sqrt{n}(v_u(\widehat{\Sigma}) - v_u(\Sigma_1)) \stackrel{d}{\rightarrow} \mathcal{N}(0, \breve{\Sigma}_1).
\end{equation}
\end{lemma}

\begin{proof}
Consider the quantity $a^T \big[ \sqrt{n}(v_u(\widehat{\Sigma}) - v_u(\Sigma_1)) \big]$ $=\sqrt{n}(a^T v_u(\widehat{\Sigma}) - a^T v_u(\Sigma_1))$ where $a \in \mathbb{R}^{k(k+1)/2}\setminus \{ 0 \}$. Note that $a^T v_u\big[(X_1-\mu_1)(X_1-\mu_1)^T\big]$, $\dots$, $a^T v_u\big[(X_n-\mu_1)(X_n-\mu_1)^T\big]$ is a sequence of i.i.d. random variables with expectation $a^T v_u(\Sigma_1)$ and variance $a^T \breve{\Sigma}_1 a$. Moreover observe that $\breve{\Sigma}_1 < \infty$ because $\breve{\Sigma}_1$ is positive definite. We can therefore apply the univariate central limit theorem to conclude that:
\begin{equation} \label{lemma_c1}
\sqrt{n}(a^T v_u(\ddot{\Sigma}_1) - a^T v_u(\Sigma_1)) \stackrel{d}{\rightarrow} \mathcal{N}(0, a^T \breve{\Sigma}_1 a),
\end{equation}
where $\ddot{\Sigma}_1 = \frac{1}{n} \sum_{i=1}^n (X_i - \mu_1) (X_i - \mu_1)^T$. We would however like to claim that:
\begin{equation}
\sqrt{n}(a^T v_u(\widehat{\Sigma}) - a^T v_u(\Sigma_1)) \stackrel{d}{\rightarrow} \mathcal{N}(0, a^T \breve{\Sigma}_1 a).
\end{equation}
In order to prove this, we use \ref{rewrite2} and set:
\begin{equation}
\sqrt{n}(a^T v_u(\widehat{\Sigma}) - a^T v_u(\Sigma_1)) = a^TA_n + a^TB_n,
\end{equation}
where we have:
\begin{equation}
\begin{aligned}
A_n &= \frac{n}{n-1} \sqrt{n} (v_u(\ddot{\Sigma}_1) - v_u(\Sigma)), \\
B_n &=  \frac{\sqrt{n}}{n-1} v_u(\Sigma_1) - \frac{n\sqrt{n}}{n-1} v_u \big[ (\widehat{\mathbb{E}}(X)-\mu_1)(\widehat{\mathbb{E}}(X)-\mu_1)^T \big].
\end{aligned}
\end{equation}
We already know from \ref{lemma_c1} that:
\begin{equation}
\sqrt{n}(a^T v_u(\ddot{\Sigma}_1) - a^T v_u(\Sigma_1)) \stackrel{d}{\rightarrow} \mathcal{N}(0, a^T \breve{\Sigma}_1 a).
\end{equation}
Therefore, so does $a^T A_n$ by Slutsky's lemma, when we view the sequence of constants $\frac{n}{n-1}$ as a sequence of random variables. For $a^T B_n$, we know that:
\begin{equation}
\sqrt{n}(a^T \widehat{\mathbb{E}}(X) - a^T \mu_1) \stackrel{d}{\rightarrow} \mathcal{N}(0, a^T \Sigma_1 a),
\end{equation}
by viewing $a^T X_1, \dots, a^T X_n$ as a sequence of random variables, noting that $E(X_1X_1^T)<\infty$ because $\breve{\Sigma}_1$ is positive definite and then applying the univariate central limit theorem. We thus have $a^T B_n \stackrel{p}{\rightarrow} 0$. We may then invoke Slutsky's lemma again for $a^TA_n + a^TB_n$ and claim that:
\begin{equation}
\sqrt{n}(a^T v_u(\widehat{\Sigma}) - a^T v_u(\Sigma_1)) \stackrel{d}{\rightarrow} \mathcal{N}(0, a^T \breve{\Sigma}_1 a).
\end{equation}
We conclude the lemma by invoking the Cramer-Wold device.

\end{proof}

\end{document}